\documentclass[11pt,reqno]{amsart}
\usepackage[a4paper]{anysize}\marginsize{2.5cm}{2.5cm}{2.5cm}{2.5cm}
\pdfpagewidth=\paperwidth \pdfpageheight=\paperheight
\allowdisplaybreaks
\usepackage{inputenc}

\usepackage{graphicx}
\usepackage{amssymb,amsmath,amstext,amsfonts,amsthm}
\usepackage{bm}

\usepackage[dvipsnames]{xcolor}

\PassOptionsToPackage{hyphens}{url}\usepackage[pdfusetitle,backref=page]{hyperref}
\usepackage{bookmark}
\hypersetup{bookmarksnumbered}
\hypersetup{colorlinks=true,
urlcolor=Black,
linkcolor=MidnightBlue,citecolor=MidnightBlue}

\renewcommand*{\backref}[1]{}
\renewcommand*{\backrefalt}[4]{%
    \ifcase #1 Not~cited.%
    \or        Cited~on~page~#2.%
    \else      Cited~on~pages~#2.%
    \fi}

\usepackage{import}
\usepackage{comment}
\usepackage{optidef}
\usepackage{epsfig}

\usepackage{color}

\usepackage{graphicx}
\usepackage{float}
\usepackage[font=scriptsize,labelfont=bf]{caption}
\usepackage{subcaption}
\graphicspath{ {./images/} } 
\usepackage[all]{xy}

\theoremstyle{plain}
\newtheorem*{theorem*}{Theorem}
\newtheorem{theorem}{Theorem}
\numberwithin{theorem}{section}

\newtheorem{lemma}[theorem]{Lemma}
\newtheorem{corollary}[theorem]{Corollary}
\newtheorem{conjecture}[theorem]{Conjecture}

\newtheorem{remark}[theorem]{Remark}

\theoremstyle{definition}

\numberwithin{equation}{section}

\newcommand{\CC}{\mathbb{C}}

\newcommand{\RR}{\mathbb{R}}

\DeclareMathOperator*{\sgn}{sgn}
\newcommand{\qq}{Q} 

\usepackage{upquote} 

\begin{document}

\title[3D genome reconstruction from partially phased Hi-C data]{3D genome reconstruction from partially\\ phased Hi-C data}

\author[D. Cifuentes, J. Draisma, O. Henriksson, A. Korchmaros, and K. Kubjas]{Diego Cifuentes, Jan Draisma, Oskar Henriksson,\\ Annachiara Korchmaros, and Kaie Kubjas}

\date{February 22, 2024}

\begin{abstract}
The 3-dimensional (3D) structure of the genome is of significant importance for many cellular processes. In this paper, we study the problem of reconstructing the 3D structure of chromosomes from Hi-C data of diploid organisms, which poses additional challenges compared to the better-studied haploid setting. With the help of techniques from algebraic geometry, we prove that a small amount of phased data is sufficient to ensure finite identifiability, both for noiseless and noisy data. In the light of these results, we propose a new 3D reconstruction method based on semidefinite programming, paired with numerical algebraic geometry and local optimization. The performance of this method is tested on several simulated datasets under different noise levels and with different amounts of phased data. We also apply it to a real dataset from mouse X chromosomes, and we are then able to recover previously known structural features.
\end{abstract}

\maketitle

\section{Introduction}
The eukaryotic chromatin has a three-dimensional (3D) structure in the cell nucleus, which has been shown to be important in regulating basic cellular functions, including gene regulation, transcription, replication,
recombination, and DNA repair~\cite{uhler2017regulation,wang2018crispr}. The 3D DNA organization is also associated with brain development and function; in particular, it is shown to be misregulated in schizophrenia ~\cite{rajarajan2018neuron,rhie2018using} and Alzheimer’s disease ~\cite{nott2019brain}.

All genetic material is stored in chromosomes, which interact in the cell nucleus, and the 3D chromatin structure influences the frequencies of such interactions.
A benchmark tool to measure such frequencies is high-throughput chromosome conformation capture (Hi-C)~\cite{lafontaine2021hi}. 
Hi-C first crosslinks cell genomes, which ``freezes'' contacts between DNA segments. Then the genome is cut into fragments, the fragments are ligated together and then are associated with equally-sized segments of the genome using high-throughput sequencing~\cite{rao20143d}. 
These segments of the genome are called \emph{loci}, and their size is known as \emph{resolution} (e.g., bins of size $1$Mb or $50$Kb). The result of Hi-C is stored in a matrix called \emph{contact matrix} whose elements are the \emph{contact counts} between pairs of loci. 

According to the structure they generate, computational methods for inferring the 3D chromatin structure from a contact matrix fall into two classes: ensemble and consensus methods. In a haploid setting (organisms having a single set of chromosomes), ensemble models such as MCMC5C~\cite{rousseau2011three}, BACH-MIX~\cite{hu2013bayesian} and Chrom3D~\cite{paulsen2017chrom3d}, try to account for structure variations on the genome across cells by inferring a population of 3D structures. On the other hand, consensus methods aim at reconstructing one single 3D structure which may be used as a model for further analysis. In this category, probability-based methods such as PASTIS \cite{varoquaux2014statistical,cauer2019inferring} and ASHIC \cite{ye2020ashic} model contact counts as Poisson random variables of the Euclidean distances between loci, and distance-based methods such as ChromSDE~\cite{zhang2013inference} and ShRec3D ~\cite{lesne20143d} model contact counts as functions of the Euclidean distances. 
An extensive overview of different 3D genome reconstruction techniques is given in~\cite{oluwadare2019overview}.

Most of the methods for 3D genome reconstructions from Hi-C data are for haploid organisms. However, like most mammals, humans are diploid organisms, in which the genetic information is stored in pairs of chromosomes called homologs. Homologous chromosomes are almost identical besides some single nucleotide polymorphisms (SNPs)~\cite{li2021understanding}. In the case of diploid organisms, the Hi-C data does not generally differentiate between homologous chromosomes.   
If we model each chromosome as a string of beads, then we associate two beads to each locus $i\in \{1,\ldots,n\}$, one bead for each homolog. Therefore, each observed contact count $c_{i,j}$ between loci $i$ and $j$
represents aggregated contacts of four different types of interactions, more precisely one of the two homologous beads associated to locus $i$ gets in contact with one of the two homologous beads associated to locus $j$, see Figure~\ref{fig:beads-loci}. This means that the Hi-C data is \emph{unphased}. \emph{Phased} Hi-C data that distinguishes contacts for homologs is rare. 
In our setting, we assume that the data is \emph{partially phased}, i.e., some of the contact counts can be associated with a homolog. For example, in the (mouse) Patski (BL6xSpretus)~\cite{deng2015bipartite,ye2020ashic} cell line, $35.6\%$ of the contact counts are phased; while this value is as low as $0.14\%$ in the human GM12878 cell line~\cite{rao20143d,ye2020ashic}. Therefore, methods for inferring diploid 3D chromatin structure need to take into account the ambiguity of diploid Hi-C data to avoid inaccurate reconstructions. 

\begin{figure}[h]
    \centering
        \includegraphics[width=0.85\textwidth]{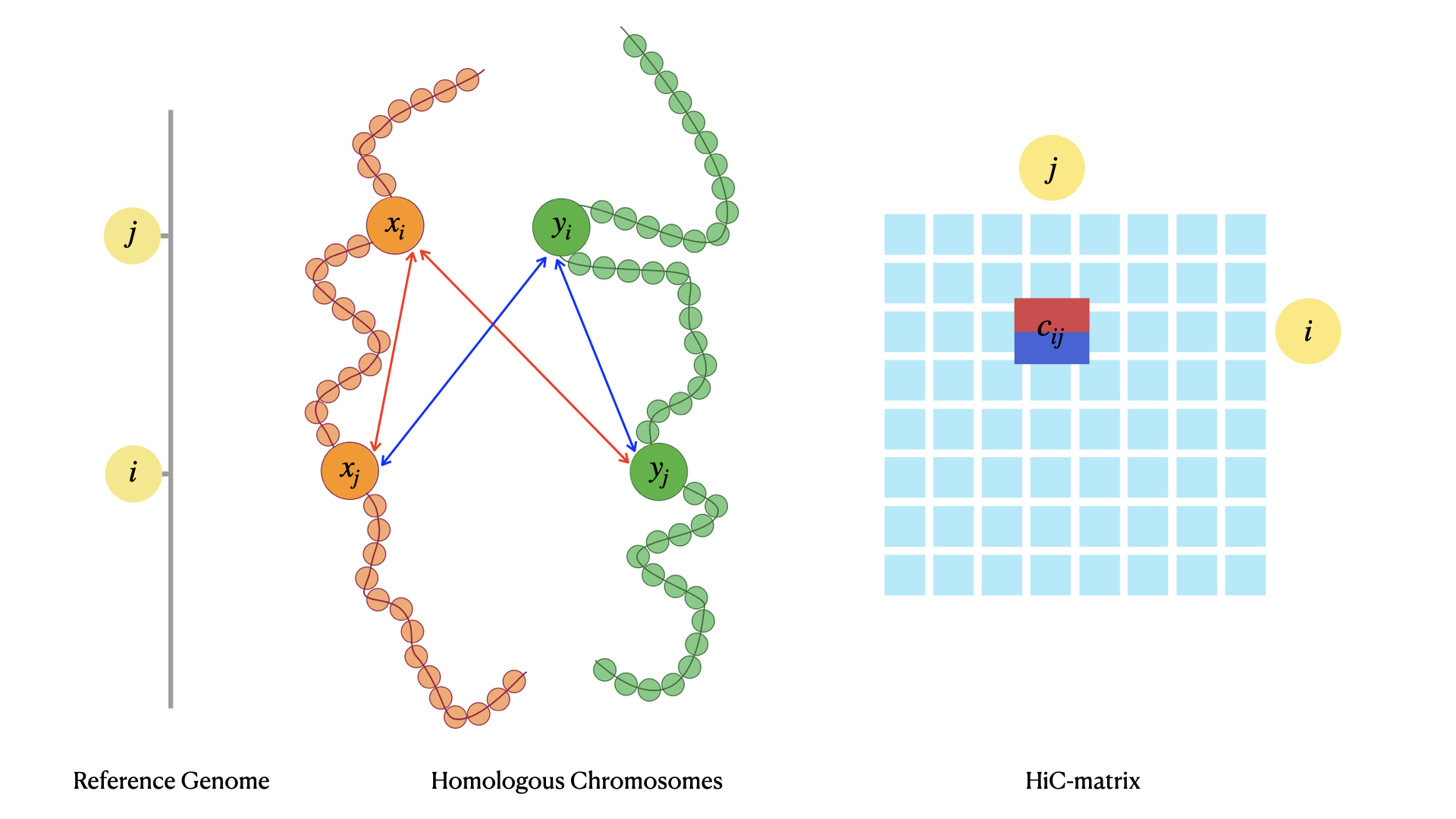}
    \caption{\textbf{Ambiguity of phased data.} Each entry $c_{i,j}$ of the Hi-C matrix corresponds to four different contacts between the two pairs $(x_i,y_i)$ for locus $i$ and $(x_j,y_j)$ for locus $j$. }
    \label{fig:beads-loci}
\end{figure}

Methods for 3D genome reconstruction in diploid organisms have been studied in~\cite{tan2018three,ye2020ashic,cauer2019inferring,luo2020hichap,belyaeva2021identifying,lindsly2021functional,segal2022can}. One approach is to phase Hi-C data~\cite{tan2018three,luo2020hichap,lindsly2021functional}, for example by assigning haplotypes to contacts based on assignments at neighboring contacts~\cite{tan2018three,lindsly2021functional}.
 Cauer et al.~\cite{cauer2019inferring} and Ye and Ma~\cite{ye2020ashic} model contact counts as Poisson random variables. To find the optimal 3D chromatin structure,  Cauer et al. maximize the associated likelihood function combined with two structural constraints. 
 The first constraint imposes that the distances between neighboring beads are similar, and the second one requires that homologous chromosomes are located in different regions of the cell nucleus. On the other hand, Ye and Ma first compute the maximum likelihood estimate of model parameters for each of the homologs separately; these estimates are then refined by estimating the distance between the homologs.
Belyaeva et
al.~\cite{belyaeva2021identifying} show identifiability of the 3D structure when the Euclidean distances between neighboring beads and higher-order contact counts between three or more loci simultaneously are given. Under these assumptions, the 3D reconstruction is obtained by combining distance geometry with semidefinite programming. Segal~\cite{segal2022can} applies recently developed imaging technology, in situ genome sequencing (IGS)~\cite{payne2021situ}, to point out issues in the assumptions made in~\cite{tan2018three,cauer2019inferring,belyaeva2021identifying}, and suggests as alternative assumptions that intra-homolog distances are smaller than corresponding inter-homolog distances and intra-homolog distances are similar for homologous chromosomes. IGS~\cite{payne2021situ} provides yet another method for inferring the 3D structure of the genome, however, at present the resolution and availability of IGS data is limited.

\subsubsection*{Contributions} In this work, we focus on a distance-based approach for partially phased Hi-C data. In particular, we assume that contacts only for some loci are phased. 
In the string of beads model, the locations of the pair of beads associated to $i$-th loci are denoted by $x_i,y_i\in \mathbb{R}^3$. Then homologs are represented  by two sequences $x_1,x_2,\ldots,x_n$ and $y_1,x_2,\ldots,y_n$ in $\mathbb{R}^3$; see Figure~\ref{fig:beads-loci}. Inferring the 3D chromatin structure corresponds to estimating the bead coordinates. Based on Lieberman-Aiden et al.~\cite{lieberman2009comprehensive}, we assume the power law dependency $c_{i,j}= \gamma d_{i,j}^{\alpha}$, where $\alpha$ is a negative conversion factor, between the distance $d_{i,j}$ and contact count $c_{i,j}$ of loci $i$ and $j$. Following Cauer et al.~\cite{cauer2019inferring}, we assume that a contact count between loci is given by the sum of all possible contact counts between the corresponding beads. We call a bead unambiguous if the contacts for the corresponding locus are phased; otherwise, we call a bead ambiguous.

Our first main contribution is to show that for negative rational conversion factors $\alpha$, knowing the locations of six unambiguous beads ensures that there are generically finitely many possible locations for the other beads, both in the noiseless (Theorem~\ref{theorem:noiseless-finite-identifiability}) and noisy (Corollary~\ref{cor:finite_identifiability_noisy}) setting. Moreover, we prove finite identifiability also in the fully ambiguous setting when $\alpha=-2$ and the number of loci is at least $12$ (Theorem~\ref{theorem:ambiguous-noiseless-finite-identifiability}).  Note that the identifiability does not hold for $\alpha=2$ as shown in~\cite{belyaeva2021identifying}.

Our second main contribution is to provide a reconstruction method when $\alpha=-2$, based on semidefinite programming combined with numerical algebraic geometry and local optimization (section~\ref{sec:reconstruction}). The general idea is the following: We first estimate the coordinates of the unambiguous beads using only the unambiguous contact counts (which precisely corresponds to the haploid setting) using the SDP-based solver implemented in ChromSDE~\cite{zhang2013inference}. We then exploit our theoretical result on finite identifiability to estimate the coordinates of the ambiguous beads, one by one, by solving several polynomial systems numerically. These estimates are then improved by a local estimation step considering all contact counts. Finally, a  clustering algorithm is used to overcome the symmetry $(x_i,y_i)\mapsto (y_i,x_i)$ in the estimation for the ambiguous beads.

The paper is organized as follows. In section~\ref{sec:math_set_up}, we introduce our mathematical model for the 3D genome reconstruction problem. In section~\ref{sec:identifiability}, we recall identifiability results in the unambiguous setting (section~\ref{subsec:euclidean_distance_geometry}) and then prove identifiability results in the partially ambiguous setting (section~\ref{subsec:partially_ambigous}) and in the fully ambiguous setting (section~\ref{subsec:identifiability_ambiguous_setting}).  We describe our reconstruction method in section~\ref{sec:reconstruction}. We test the performance of our method on synthetic datasets and on a real dataset from the mouse X chromosomes in section~\ref{sec:experiments}. We conclude with a discussion about future research directions in section~\ref{discussion}.

\section{Mathematical model for 3D genome reconstruction}~\label{sec:math_set_up}

In this section we introduce the distance-based model under which we study 3D genome reconstruction. In section~\ref{subsec:contact_count_matrices} we give the background on contact count matrices. In section~\ref{subsec:distances} we describe a power-law between contacts and distances, which allows to translate the information about contacts into distances.

\subsection{Contact count matrices} \label{subsec:contact_count_matrices}

We model the genome as a string of $2n$ beads, corresponding to $n$ pairs of homologous beads. The positions of the beads are recorded by a matrix 
\[Z=[x_1,\ldots,x_n,y_1,\ldots,y_n]^T \in \RR^{2n \times 3}.\] 
The positions $x_i$ and $y_i$ correspond to homologous beads. When convenient, we use the notation $z_1:=x_1,\ldots,z_n:=x_n,z_{n+1}:=y_1,\ldots,z_{2n}:=y_n$. In this notation, 
\[Z=[z_1,\ldots,z_n,z_{n+1},\ldots,z_{2n}]^T \in \RR^{2n \times 3}.\] Let $U$ be the subset of pairs that are unambiguous, i.e., beads in the pair can be distinguished, and let $A$ be the subset of pairs that are ambiguous, i.e.,  beads in the pair cannot be distinguished. The sets $U$ and $A$ form a partition of $[n]$. 

A Hi-C matrix $C$ is a matrix with each row and column corresponding to a genomic locus. Following Cauer et al.~\cite{cauer2019inferring}, we call these contact counts ambiguous and denote the corresponding contact count matrix by $C^A$. If parental genotypes are available, then one can use SNPs to map some reads to each haplotype~\cite{deng2015bipartite,minajigi2015comprehensive,rao20143d}. If both ends of a read contains SNPs that can be associated to a single parent, then the contact count is called unambiguous and the corresponding contact count matrix is denoted by $C^U$. Finally, if only one of the genomic loci present in an interaction can be mapped to one of the homologous chromosomes, then the count is called partially ambiguous and the contact count matrix is denoted by $C^P$.

The unambiguous count matrix $C^U$ is a $2n \times 2n$ matrix with the first $n$ indices corresponding to $x_1,\ldots,x_n$ and the last $n$ indices corresponding to $y_1,\ldots,y_n$. The ambiguous count matrix $C^A$ is an $n \times n$ matrix and we assume that each ambiguous count is the sum of four unambiguous counts:
$$
c^A_{i,j} =  c^U_{i,j}+c^U_{i,j+n}+c^U_{i+n,j}+c^U_{i+n,j+n}.
$$
The partially ambiguous count matrix $C^P$ is a $2n \times n$ matrix and each partially ambiguous count is the sum of two unambiguous counts:
$$
c^P_{i,j} = c^U_{i,j} + c^U_{i,j+n}.
$$

\begin{figure}[ht]
\centering
\vspace{1em}
\begin{subfigure}{0.24\textwidth}
\centering
\scriptsize
\def\svgwidth{0.8\linewidth}
\includegraphics[width=\svgwidth]{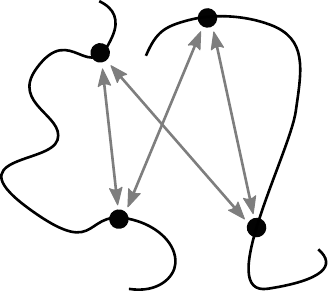}
\caption{$c^A_{i,j}$ for $i,j\in A$}
\end{subfigure}
\hspace{0.2em}
\begin{subfigure}{0.24\textwidth}
\centering
\scriptsize
\def\svgwidth{0.8\linewidth}
\includegraphics[width=\svgwidth]{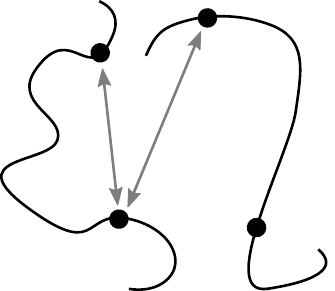}
\caption{$c^P_{i,j}$ for $i\in U,j\in A$}
\end{subfigure}
\hspace{0.2em}
\begin{subfigure}{0.24\textwidth}
\centering
 \scriptsize
 \def\svgwidth{0.8\linewidth}
\includegraphics[width=\svgwidth]{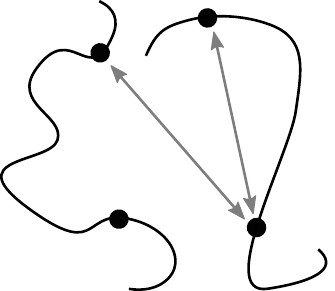}
\caption{$c^P_{i+n,j}$ for $i\in U,j\in A$}
\label{fig:conact_counts}
\end{subfigure}

\vspace{2em}

\begin{subfigure}{0.23\textwidth}
\centering
 \scriptsize
 \def\svgwidth{0.8\linewidth}
\includegraphics[width=\svgwidth]{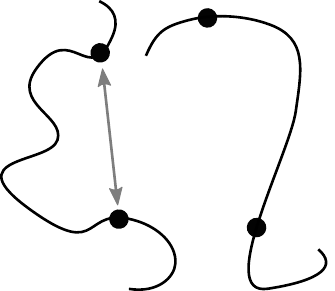}
\caption{$c^U_{i,j}$ for $i,j\in U$}
\end{subfigure}
\hspace{0.2em}
\begin{subfigure}{0.23\textwidth}
\centering
 \scriptsize
 \def\svgwidth{0.8\linewidth}
\includegraphics[width=\svgwidth]{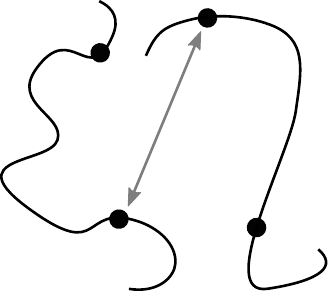}
\caption{$c^U_{i,j+n}$ for $i,j\in U$}
\end{subfigure}
\hspace{0.2em}
\begin{subfigure}{0.23\textwidth}
\centering
 \scriptsize
 \def\svgwidth{0.8\linewidth}
\includegraphics[width=\svgwidth]{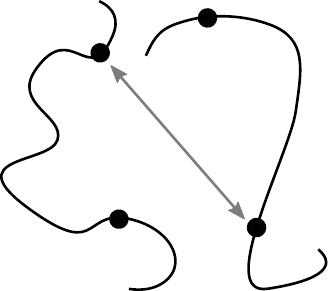}
\caption{$c^U_{i+n,j}$ for $i,j\in U$}
\end{subfigure}
\hspace{0.2em}
\begin{subfigure}{0.23\textwidth}
\centering
 \scriptsize
 \def\svgwidth{0.8\linewidth}
\includegraphics[width=\svgwidth]{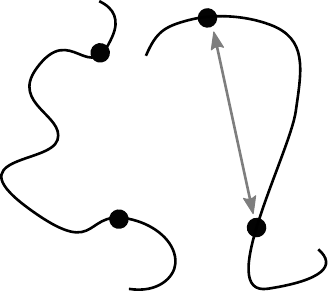}
\caption{$c^U_{i+n,j+n}$ for $i,j\in U$}
\end{subfigure}

\caption{Seven different types of contacts between the $i$th and $j$th locus.}

\end{figure}

\subsection{Contacts and distances} \label{subsec:distances}

Denoting the distance $\| z_i -z_j \|$ between $z_i$ and $z_j$ by $d_{i,j}$, the power law dependency observed by Lieberman-Aiden et al.~\cite{lieberman2009comprehensive} can be written as
\begin{equation} \label{eqn: relation between contacts and distances}
c^U_{i,j} = \gamma d_{i,j}^{\alpha},  
\end{equation}
where $\alpha<0$ is a conversion factor and $\gamma>0$ is a scaling factor. This relationship between contact counts and distances is assumed in~\cite{belyaeva2021identifying,zhang2013inference}, while in~\cite{cauer2019inferring,varoquaux2014statistical} the contact counts $c_{i,j}$ are modeled as Poisson random variables with the Poisson parameter being $\beta d_{i,j}^{\alpha}$.  

In our paper, we assume that contact counts are related to distances by~(\ref{eqn: relation between contacts and distances}). Similarly to~\cite{belyaeva2021identifying}, we set $\gamma=1$ and in parts of the article  $\alpha=-2$. In general, the conversion factor $\alpha$ depends on a dataset and its estimation can be part of the reconstruction problem~\cite{varoquaux2014statistical,zhang2013inference}. Setting $\gamma=1$ means that we recover the configuration up to a scaling factor. In practice, the configuration can be rescaled using biological knowledge, e.g., the radius of the nucleus. 

Our approach to 3D genome reconstruction builds on the power law dependency between contacts and distances between unambiguous beads. We convert the empirical contact counts to Euclidean distances and then aim to reconstruct the positions of beads from the distances. This leads us to the following system of equations:
\begin{equation}\label{eq:full_system}
\hspace{-0.25cm}
\begin{cases}
c^A_{i,j} = \|x_i-x_j\|^{\alpha} +  \|x_i-y_j\|^{\alpha} +  \|y_i-x_j\|^{\alpha} +  \|y_i-y_j\|^{\alpha}  &\hspace{-5pt} \forall i, j \in A \\
c^P_{i,j} = \|x_i-x_j\|^{\alpha} \hspace{-2pt} + \hspace{-2pt}   \|x_i-y_j\|^{\alpha},\,\,  c^P_{i+n,j} = \|y_i-x_j\|^{\alpha} \hspace{-2pt} + \hspace{-2pt}  \|y_i-y_j\|^{\alpha} & \hspace{-5pt}  \forall i \in U, j \in A \\
c^U_{i,j} =  \|x_i-x_j\|^{\alpha}, \,\, c^U_{i,j+n} =  \|x_i-y_j\|^{\alpha}, 
& \\
c^U_{i+n,j} =  \|y_i-x_j\|^{\alpha}, \,\, c^U_{i+n,j+n} =  \|y_i-y_j\|^{\alpha} &\hspace{-5pt} \forall i, j \in U
\end{cases}
\end{equation}

If $\alpha$ is an even integer, then~(\ref{eq:full_system}) is a system of rational equations.

Determining the points $x_i, y_i$, where $i\in U$, is the classical Euclidean distance problem: We know the (noisy) pairwise distances between points and would like to construct the locations of points, see section~\ref{subsec:euclidean_distance_geometry} for details. Hence after section~\ref{subsec:euclidean_distance_geometry} we assume that we have estimated the locations of points $x_i, y_i$, where $i\in U$, and we would like to determine the points $x_i, y_i$, where $i\in A$.

\section{Identifiability} \label{sec:identifiability}

In this section, we study the uniqueness of the solutions of the system~(\ref{eq:full_system}) up to rigid transformations (translations, rotations and reflections), or in other words, the identifiability of the locations of beads. We study the unambiguous, partially ambiguous and ambiguous settings in sections~\ref{subsec:euclidean_distance_geometry},~\ref{subsec:partially_ambigous} and~\ref{subsec:identifiability_ambiguous_setting}, respectively.

\subsection{Unambiguous setting and Euclidean distance geometry} \label{subsec:euclidean_distance_geometry}

If all pairs are unambiguous, i.e., $U=[n]$, then constructing the original points translates to a classical problem in Euclidean distance geometry. The principal task in Euclidean distance geometry is to construct original points from pairwise distances between them.  In the rest of the subsection, we will recall how to solve this problem. Since pairwise distances are invariant under translations, rotations and reflections (rigid transformations), then the original points can be reconstructed up to rigid transformations. For an overview of distance geometry and Euclidean distance matrices, we refer the reader to~\cite{dokmanic2015euclidean,krislock2012euclidean,liberti2014euclidean,mucherino2012distance}.

The Gram matrix of the points $z_1,\ldots,z_{2n}$ is defined as
$$
G = Z Z^T =  [z_1,\ldots,z_{2n}]^T \cdot [z_1,\ldots,z_{2n}] \in \RR^{2n \times 2n}.
$$
Let $\overline{z} = \frac{1}{2n} \sum_{i=1}^{2n} z_i$ and   $\tilde{z}_i= z_i - \overline{z}$ for $i=1,\ldots,2n$. The matrix $\tilde{Z} = [\tilde{z}_1,\ldots, \tilde{z}_{2n}]^T$ gives the locations of points after centering them around the origin. Let $\tilde{G}$ denote the Gram matrix of the centered point configuration $\tilde{z}_1,\ldots, \tilde{z}_{2n}$.

Let $D_{i,j} = \|z_i - z_j\|^2$ denote the squared Euclidean distance between the points $z_i$ and $z_j$. The Euclidean distance matrix of the points $z_1,\ldots,z_{2n}$ is defined as $D=(D_{i,j})_{1 \leq i,j \leq 2n} \in \RR^{2n \times 2n}$.  To express the centered Gram matrix in terms of the Euclidean distance matrix, we define the geometric centering matrix
$$
J=I_{2n} - \frac{1}{2n} \bm{1} \bm{1}^T,
$$
where $I_{2n}$ is the $2n \times 2n$ identity matrix and $\bm{1}$ is the vector of ones.
The linear relationship between $\tilde{G}$ and $D$ is given by
$$
\tilde{G} = -\frac{1}{2} JDJ.
$$
Therefore, given the Euclidean distance matrix, we can construct the centered Gram matrix for the points $z_1,\ldots,z_{2n}$.

The centered points up to rigid transformations are extracted from the centered Gram matrix $\tilde{G}$ using the eigendecomposition $\tilde{G}=Q \Lambda Q^{-1}$, where $Q$ is orthonormal  and $\Lambda$ is a  diagonal matrix with entries ordered in decreasing order $\lambda_1 \geq \lambda_2 \geq \ldots \geq \lambda_{2n} \geq 0$. We define $\Lambda_3^{1/2} := [\text{diag}(\sqrt{\lambda_1},\sqrt{\lambda_2}, \sqrt{\lambda_3}),\bm{0}_{3 \times (2n-3)}]^T$ and set $\hat{Z} = Q \Lambda_3^{1/2}$. In the case of noiseless distance matrix $D$, the Gram matrix $\tilde{G}$ has rank three and the diagonal matrix $\Lambda$ has precisely three non-zero entries. Hence we could obtain $\hat{Z}$ also from $Q \Lambda^{1/2}$ by truncating zero columns. Using $\Lambda_3^{1/2}$ has the advantage that it gives an approximation for the points also for a noisy distance matrix $D$. The uniqueness of $\hat{Z}$ up to rotations and reflections follows from~\cite[Proposition 3.2]{krislock2010semidefinite}, which states that $AA^T = BB^T$ if and only if $A=BQ$ for some orthogonal matrix $Q$.

The procedure that transforms the distance matrix to origin centered Gram matrix and then uses eigendecomposition for constructing original points is called classical multidimensional scaling (cMDS)~\cite{cox2008multidimensional}. Although cMDS is widely used in practice, it does not always find the distance matrix that minimizes the Frobenius norm to the empirical noisy distance matrix~\cite{sonthalia2021can}. 
Other approaches to solving the Euclidean distance and Euclidean completion problems include non-convex~\cite{fang2012euclidean,mishra2011low} as well semidefinite formulations~\cite{alfakih1999solving,fazel2003log,nie2009sum,weinberger2007graph,zhang2013inference,cayton2006robust}.

\subsection{Partially ambiguous setting}\label{subsec:partially_ambigous}

The next theorem establishes the uniqueness of the solutions of the system~(\ref{eq:full_system}) in the presence of ambiguous pairs. In particular, it states that there are finitely many possible locations for beads in one ambiguous pair given the locations of six unambiguous beads. The identifiability results in this subsection hold for all negative rational numbers $\alpha$. In the rest of the paper, we denote the true but unknown coordinates by $x^*$ and the symbol $x$ stands for a variable that we want to solve for. We write $\|\cdot\|$ for the standard inner product on $\mathbb{R}^3$.

\begin{theorem} \label{theorem:noiseless-finite-identifiability}
Let $\alpha$ be a negative rational number. Then for
$a^*,b^*,\ldots,f^*,x^*,\\ y^* \in \mathbb{R}^3$ sufficiently general, the system
of six equations
\begin{equation} \label{eqn:partial_contact_alpha}
 \|x-t^*\|^\alpha + \|y-t^*\|^\alpha = \|x^*-t^*\|^\alpha +
\|y^*-t^*\|^\alpha \text{ for } t^*=a^*,b^*,\ldots,f^* 
\end{equation}
in the six unknowns $x_1,x_2,x_3,y_1,y_2,y_3 \in \mathbb{R}$ has only finitely
many solutions.
\end{theorem}

\begin{remark}
The proof will show that this system has only
finitely many solutions over the {\em complex} numbers.

We believe that the theorem holds for general nonzero rational $\alpha$.
Indeed, our argument works, with a minor modification, also for
$\alpha>2$, but for $\alpha$ in the range $(0,2]$ a refinement of the
argument is needed.
\end{remark}

\begin{proof}
First write $\qq(x):=x_1^2 + x_2^2 + x_3^2$, so that $\|x\|=\sqrt{\qq(x)}$
for $x \in \mathbb{R}^3$. The advantage of $\qq$ over $\|x\|$ is that it is
well-defined on $\CC^3$.

Write $\frac{\alpha}{2}=\frac{m}{n}$ with $m,n$ relatively prime integers, $m \neq
0$, and $n>0$. Consider the affine variety $X \subseteq (\CC^3)^8 \times
(\CC^2)^6$ consisting of all tuples
\[ 
((a^*,\ldots,f^*,x^*,y^*),(r_{t^*},s_{t^*})_{t^*=a^*,\ldots,f^*})
\]
such that 
\[ \qq(x^*-t^*)^m = r_{t^*}^n \neq 0 \text{ and } \qq(y^*-t^*)^m = s_{t^*}^n \neq
0  \text{ for } t^*=a^*,\ldots,f^*.\]
Note that, if $x^*,t^*$ are real, then it follows that
\[ \qq(x^*-t^*)^m = (\|x^*-t^*\|^{\alpha})^n, \]
and similarly for $\qq(y^*-t^*)$. Hence if
$a^*,\ldots,y^*$ are all real, then the point
\begin{equation} \label{eq:Real}
((a^*,\ldots,f^*,x^*,y^*),(\|x^*-t^*\|^\alpha,\|y^*-t^*\|^\alpha)_{t^*})
\end{equation}
is a point in $X$ with real-valued coordinates. 

The projection $\pi$ from $X$ to the open affine subset $U \subseteq
(\CC^3)^8$ where all $\qq(x^*-t^*)$ and $\qq(y^*-t^*)$ are nonzero is a finite
morphism with fibers of cardinality $n^{12}$; to see this cardinality note
that there are $n$ possible choices for each of the numbers $r_{t^*},
s_{t^*}$. Each irreducible component of $X$ is a smooth variety of
dimension $24$.

Consider the map $\psi:X \to (\CC^3 \times \CC^1)^6$ defined by
\[ ((a^*,\ldots,f^*,x^*,y^*),(r_{t^*},s_{t^*})_{t^*})
\mapsto ((t^*,r_{t^*}+s_{t^*}))_{t^*} \]
We claim that for $q$ in some open dense subset of $X$, the derivative
$d_q \psi$ has full rank $24$. For this, it suffices to find one point
$p \in U$ such that $d_q \psi$ has rank $24$ at each of the $n^{12}$
points $q \in \pi^{-1}(p)$. We take a real-valued point 
$p:=(a^*,b^*,\ldots,f^*,x^*,y^*) \in (\mathbb{R}^3)^8$ 
to be specified later on. 
Let $q \in \pi^{-1}(p)$. Then, near $q$, the map $\psi$
factorises via $\pi$ and the unique algebraic map $\psi':U \to (\CC^3
\times \CC^1)^6$ (defined near $p$) which on a neighborhood of $p$ in 
$U \cap (\mathbb{R}^3)^8$ equals
\[ \psi'(a,\ldots,f,x,y)=((t,\xi_{t^*} \cdot \qq(x-t)^{\alpha/2} +
\eta_{t^*} \cdot 
\qq(y-t)^{\alpha/2}))_{t=a,\ldots,f}  \in (\CC^3 \times \CC^1)^6 \]
where $\xi_{t^*}$ and $\zeta_{t^*}$ are $n$-th roots of unity in
$\CC$ depending on which $q$ is chosen among the $n^{12}$
points in $\pi^{-1}(p)$. The situation is summarised in the
following diagram:
\[ 
\xymatrix{ 
(X,q) \ar[d]_\pi \ar[drr]^\psi & & \\
(U,p) \ar[rr]_{\psi'} & & ((\CC^3 \times \CC^1)^6,\psi(q)).
}
\]
Now, $d_q \psi = d_p \psi' \circ d_q \pi$, and since $d_q \pi$ is
a linear isomorphism, it suffices to prove that $d_p \psi'$ is a linear
isomorphism. Suppose that $(a',\ldots,f',x',y') \in \ker d_p \psi'$.
Then, since the map $\psi'$ remembers $a,\ldots,f$, it follows immediately
that $a'=\ldots=f'=0$. On the other hand, by differentiating we find that,
for each $t^* \in \{a^*,\ldots,f^*\}$,
\begin{align*}
&\xi_{t^*} \cdot (\alpha/2) \cdot \qq(x^*-t^*)^{\alpha/2-1}
\cdot 2 \cdot \langle x',x^*-t^* \rangle \\
+ &\eta_{t^*} \cdot (\alpha/2) \cdot \qq(y^*-t^*)^{\alpha/2-1}
\cdot 2 \cdot \langle y',y^*-t^* \rangle = 0,
\end{align*}
where $\langle \cdot, \cdot \rangle$
stands for the standard bilinear form on $\CC^3$. In other words, the
vector $(x',y') \in \CC^6$ is in the kernel of the $6 \times 6$-matrix
\[ M:=\begin{bmatrix} 
\|x^*-a^*\|^{\alpha-2} \cdot \xi_{a^*} \cdot (x^*-a^*) &
\|y^*-a^*\|^{\alpha-2} \cdot \eta_{a^*} \cdot (y^*-a^*) \\
\vdots & \vdots \\
\|x^*-f^*\|^{\alpha-2} \cdot \xi_{f^*} \cdot (x^*-f^*) & 
\|y^*-f^*\|^{\alpha-2} \cdot \eta_{f^*} \cdot (y^*-f^*) 
\end{bmatrix} 
\]
where we have interpreted $a^*,\ldots,f^*,x^*,y^*$ as row
vectors.  It suffices to show that, for some specific choice
of $p=(a^*,\ldots,f^*,x^*,y^*) \in (\mathbb{R}^3)^8$,
this matrix is nonsingular {\em for all $n^{12}$ choices of
$((\xi_{t^*},\eta_{t^*}))_{t^*}$}.

We choose $a^*,\ldots,f^*,x^*,y^*$ as the vertices of the unit
cube, as follows: 
\begin{align*}
a^*&=(1,0,0) & b^*&=(0,1,0) & c^*&=(0,0,1) \\
c^*&=(0,1,1) & d^*&=(1,0,1) & f^*&=(1,1,0) \\
x^*&=(0,0,0) & y^*&=(1,1,1).
\end{align*}
Then the matrix $M$ becomes, with $\beta=\alpha-2$:
\[
\begin{bmatrix}
-\xi_{a^*}& 0& 0& 0& 2^{\frac{\beta}{2}}\cdot\eta_{a^*}&
2^{\frac{\beta}{2}}\cdot\eta_{a^*}\\ 
0& -\xi_{b^*}& 0& 2^{\frac{\beta}{2}}\cdot\eta_{b^*}& 0&
2^{\frac{\beta}{2}}\cdot\eta_{b^*}\\
0& 0& -\xi_{c^*}& 2^{\frac{\beta}{2}}\cdot\eta_{c^*}&
2^{\frac{\beta}{2}}\cdot\eta_{c^*}& 0\\ 
0& -(2^{\frac{\beta}{2}}\cdot\xi_{d^*})&
-(2^{\frac{\beta}{2}}\cdot\xi_{d^*})& \eta_{d^*}& 0& 0\\
-(2^{\frac{\beta}{2}}\cdot\xi_{e^*})& 0&
-(2^{\frac{\beta}{2}}\cdot\xi_{e^*})& 0& \eta_{e^*}& 0\\ 
-(2^{\frac{\beta}{2}}\cdot\xi_{f^*})&
-(2^{\frac{\beta}{2}}\cdot\xi_{f^*})& 0& 0& 0& \eta_{f^*}
\end{bmatrix}. 
\]
Now, $\det(M)$ equals 
\begin{equation} \label{eq:Det} -\xi_{a^*} \cdot \xi_{b^*} \cdot \xi_{c^*} \cdot
\eta_{d^*} \cdot \eta_{e^*} \cdot
\eta_{f^*} + 2^{2+3\beta} \cdot \eta_{a^*} \cdot \eta_{b^*}
\cdot\eta_{c^*} \cdot \xi_{d^*} \cdot
\xi_{e^*} \cdot \xi_{f^*} + 2^{2\beta} \cdot R 
\end{equation}
where $R$ is a sum of (products of) roots of unity. Now $\alpha<0$ implies
that $\beta<-2$, so that $2+3\beta<2\beta<0$. Since roots of unity have
$2$-adic valuation $0$, the second term in the expression above is the
unique term with minimal $2$-adic valuation. Hence $\det(M)
\neq 0$, as desired. 

It follows that $\psi$ is a dominant morphism from each irreducible
component of $X$ into $(\CC^3 \times \CC^1)^6$, and hence for all
$q$ in an open dense subset of $X$, the fiber $\psi^{-1}(\psi(q))$ is
finite. This then holds, in particular, for $q$ in an open dense subset
of the real points as in \eqref{eq:Real}. This proves the theorem.
\end{proof}

\begin{remark}
If $\alpha>2$, then $\beta>0$, and hence the unique term
with minimal $2$-adic valuation in \eqref{eq:Det} is the
first term. This can be used to show that the theorem holds then, as well. The only subtlety is that for positive $\alpha$, solutions where $x$ or $y$ equal one of the points $a^*,\ldots,f^*$ are not automatically excluded, and these are not seen by the variety $X$. But a straightforward argument shows that such solutions do not exist for sufficiently general choices of $a^*,\ldots,f^*,x^*,y^*$.
\end{remark}

We now consider the setting when we know locations of seven unambiguous beads. In the special case when $\alpha=-2$, we construct the ideal generated by the polynomials obtained from rational equations~(\ref{eqn:partial_contact_alpha}) for seven unambiguous beads after moving all terms to one side and clearing the denominators. Based on symbolic computations in \texttt{Macaulay2} for the degree of this ideal, we conjecture that the location of a seventh unambiguous bead guarantees unique identifiability of an ambiguous pair of beads:

\begin{conjecture}\label{conj:unique_identifiability_alpha=2}
Let $a^*,b^*,c^*,d^*,e^*,f^*,g^*,x^*,y^* \in \mathbb{R}^3$ be sufficiently general. 
The system of rational equations
\begin{equation}\label{poly_eq2}
   \frac{1}{\|t^*- x^*\|^2} + \frac{1}{\|t^* - y^*\|^2}=\frac{1}{\|t^* - x\|^2} + \frac{1}{\|t^* - y\|^2} \text{ for } t^*=a^*,b^*,c^*,d^*,e^*,f^*,g^*
\end{equation}
has precisely two solutions $(x^*,y^*)$ and $(y^*,x^*)$.
\end{conjecture} 

In practice, we only have noisy estimates  $a,b,\ldots,f \in \RR^3$ of the true positions of unambiguous beads $a^*,b^*,\ldots,f^* \in \RR^3$,
and we have noisy observations $c_t$ of the true contact counts $c_t^* := \|x^*-t^*\|^{\alpha}+\|y^*-t^*\|^{\alpha}$.
We aim to find $x,y \in \RR^3$ such that
\begin{equation*}
   \|x-t\|^{\alpha}+\|y-t\|^{\alpha} = c_t \text{ for } t=a,b,\ldots,f.
\end{equation*}
We may write $c_t = \|x^*-t\|^{\alpha}+\|y^*-t\|^{\alpha}+\epsilon_{t}$
for some $\epsilon_t$ that depends on the noise level.
Hence, the above system of equations can be rephrased as
\begin{equation}\label{rational_eq_noisy3}
   \|x-t\|^{\alpha}+\|y-t\|^{\alpha} = \|x^*-t\|^\alpha +
\|y^*-t\|^\alpha + \epsilon_{t} \text{ for } t=a,b,\ldots,f.
\end{equation}
In the following corollary we show that this system has generically finitely many solutions. 

\begin{corollary} \label{cor:finite_identifiability_noisy}
Let $\alpha$ be a negative rational number. Then for $a,b,\ldots,f,x^*,\\y^* \in \mathbb{R}^3$ and $\epsilon_{a},\epsilon_{b},\ldots,\epsilon_{f} \in \RR$ sufficiently general, the system of six equations
\begin{equation}\label{rational_eq_noisy}
    \|x-t\|^{\alpha}+\|y-t\|^{\alpha} = \|x^*-t\|^\alpha +
\|y^*-t\|^\alpha + \epsilon_{t} \text{ for } t=a,b,\ldots,f
\end{equation}
in the six unknowns $x_1,x_2,x_3,y_1,y_2,y_3 \in \mathbb{R}$ has only finitely
many solutions.
\end{corollary}

\begin{proof}
Recall the map $\psi:X \to (\CC^3 \times \CC^1)^6$ from the proof of Theorem~\ref{theorem:noiseless-finite-identifiability} defined by
\[ ((a,\ldots,f,x^*,y^*),(r_{x^*,t},s_{y^*,t})_{t})
\mapsto ((t,r_{x^*,t}+s_{y^*,t}))_{t}. \]
We showed that $\psi$ is a dominant morphism from each irreducible
component of $X$ into $(\CC^3 \times \CC^1)^6$, and that each irreducible component of $X$ is 24-dimensional. Every solution to \eqref{rational_eq_noisy} is the $(x,y)$-component of a point in the fiber
\[ \psi^{-1}((t,||x^*-t||^\alpha+||y^*-t||^\alpha+\epsilon_t))_t.\]
Since this is a fiber over a sufficiently general point, the fiber is finite.
\end{proof}

Corollary~\ref{cor:finite_identifiability_noisy} will be the basis of a numerical algebraic geometric based reconstruction method in section~\ref{sec:reconstruction}.

\subsection{Ambiguous setting} \label{subsec:identifiability_ambiguous_setting}

Finally we consider the ambiguous setting, where one would like to reconstruct the locations of beads only from ambiguous contact counts. It is shown in~\cite{belyaeva2021identifying} that for $\alpha=2$, one does not have finite identifiability no matter how many pairs of ambiguous beads one considers. We show finite identifiability for the locations of beads  given contact counts for $12$ pairs of ambiguous beads for $\alpha=-2$ in both the noisy and noiseless setting. We believe that the result might be true for further conversion factors $\alpha$'s, however our proof technique does not directly generalize. 

\begin{theorem} \label{theorem:ambiguous-noiseless-finite-identifiability}
Let $\alpha=-2$. Then for
$(c_{ij})_{1\leq i<j\leq 12}\in \mathbb{R}^{66}$ sufficiently general, the system
of $66$ equations
\begin{equation}\label{eqn:ambiguous-system}
\begin{aligned}
 &\|x_i-x_j\|^\alpha + \|x_i-y_j\|^\alpha + \|y_i-x_j\|^\alpha + \|y_i-y_j\|^\alpha = c_{ij} \text{ for } 1 \leq i<j \leq 12 
\end{aligned}
\end{equation}
in the $72$ unknowns $x_{1,1},x_{1,2},x_{1,3},y_{1,1},y_{1,2},y_{1,3},\ldots,x_{12,1},x_{12,2},x_{12,3},y_{12,1},y_{12,2},\\y_{12,3} \in \mathbb{R}$ has only finitely
many solutions up to rigid transformations. In particular, it holds that for sufficiently general $(x_1^*,y_1^*,\ldots,x_{12}^*,y_{12}^*)\in(\RR^3)^{24}$, the system 
\begin{equation}\label{eqn:ambiguous-system_exact}
\begin{aligned}
 &\|x_i-x_j\|^\alpha + \|x_i-y_j\|^\alpha + \|y_i-x_j\|^\alpha + \|y_i-y_j\|^\alpha = \\
 &\|x_i^*-x_j^*\|^\alpha + \|x_i^*-y_j^*\|^\alpha + \|y_i^*-x_j^*\|^\alpha + \|y_i^*-y_j^*\|^\alpha \text{ for } 1 \leq i<j \leq 12 
\end{aligned}
\end{equation}
has finitely many solutions up to rigid transformation.
\end{theorem}

\begin{proof}
As before, we write $\qq(x):=x_1^2 + x_2^2 + x_3^2$, so that $\|x\|=\sqrt{\qq(x)}$
for $x \in \mathbb{R}^3$. Consider the affine open subset $X \subseteq (\CC^3)^{24}$ consisting of all tuples $
(x_1^*,y_1^*,\ldots,x_{12}^*,y_{12}^*)$
such that 
\[ \qq(x_i^*-x_j^*) \neq 0,\:\: \qq(x_i^*-y_j^*) \neq 0,\:\: \qq(y_i^*-x_j^*)\neq 0  \:\:\text{and}\:\: \qq(y_i^*-y_j^*) \neq
0  \text{ for } i < j.\]
Consider also the map $\psi\colon X \to \CC^{66}$ defined by
\begin{align*}
 (x_1^*,\ldots,y_{12}^*)
\mapsto\left(\qq(x_i^*\hspace{-1pt}-x_j^*)^{-1}{\hspace{-3pt}}+\qq(x_i^*\hspace{-1pt}-y_j^*)^{-1}{\hspace{-3pt}}+\qq(y_i^*\hspace{-1pt}-x_j^*)^{-1}{\hspace{-3pt}}+\qq(y_i^*\hspace{-1pt}-y_j^*)^{-1}\right)_{i<j}. 
\end{align*}
By a computer calculation (with exact arithmetic) we found that at a randomly chosen $q \in X$ with rational coordinates, the derivative $d_q \psi$ had full rank $66$. 
It then follows that for $q$ in some open dense subset of $X$, 
$d_q \psi$ has rank $66$. Hence $\psi$ is dominant, and for any sufficiently general $c \in \CC^{66}$, all irreducible components of the fiber $\psi^{-1}(c)$ have dimension $6$. Moreover, each such component $C$ is preserved by the $6$-dimensional connected group $G=SO(3,\CC) \ltimes\CC^3$.

The stabilizer in $G$ of a sufficiently general point in $X$ is zero-dimensional. This follows from a Lie algebra argument: if a point $(x_1^*,y_1^*,\ldots,x_{12}^*,y_{12}^*) \in X$ has a positive-dimensional stabilizer in $G$, then there is a nonzero element $A$ in the Lie algebra of $SO(3,\CC)$ that maps all the differences $x_i^*-x_j^*,x_i^*-y_j^*,y_i^*-y_j^*$ to zero. Since $A$ is a skew-symmetric matrix and hence of rank 2, it follows that all points $x_i^*,y_j^*$ lie on a line. The variety of such collinear tuples has dimension 28, so it does not map dominantly to $\CC^{66}$. Hence there exists a Zariski open dense subset $V\subseteq\CC^{66}$ such that for all $c\in V$, the fiber $\psi^{-1}(c)$ contains no points with positive-dimensional stabilizers in $G$, and hence $\psi^{-1}(c)$ is a disjoint union of finitely many 6-dimensional $G$-orbits.
Likewise, $\psi^{-1}(V)$ is a Zariski open dense subset of $(\CC^3)^{24}$ such that $\psi^{-1}(\psi(q))$ consists of finitely many $G$-orbits for all $q\in\psi^{-1}(V)$. With this, we have proven the complex analog of the theorem. 

To obtain the statement over the real numbers, we note that if $c\in V$ has real-valued coordinates, then a finite number of the $G$-orbits that make up $\psi^{-1}(c)$  contain a real-valued tuple. If $G\cdot q$ for $q\in(\RR^3)^{24}$ is such an orbit, it holds that $(G\cdot q)\cap(\mathbb{R}^3)^{24}=(SO(3,\mathbb{R}) \ltimes \mathbb{R}^3) \cdot q$ whenever the 24 points that make up the tuple $q$ are not coplanar. The set of coplanar configurations form a subset of $X$ of dimension 51, and does therefore not map dominantly to $\CC^{66}$. Hence, by shrinking $V$ appropriately, we can assume that no fibers above it contain coplanar configurations. In particular, this means that the real part of the fiber over any real point in $V$ consists of a finitely many orbits under the action of $SO(3,\mathbb{R}) \ltimes \mathbb{R}^3$, as desired. 
\end{proof}

\begin{remark}
A standard numerical algebraic geometry computation with monodromy and the certification techniques of \cite{breiding2023certifying},
 using \texttt{HomotopyContinuation.jl} (see, e.g., \cite{sturmfels2021likelihood}), proves that the system \eqref{rational_eq_noisy} generically has more than 1000 complex solutions up to the action of $O(3,\mathbb{C})\ltimes\mathbb{C}^3$ and the symmetries $(x_i,y_i)\mapsto (y_i,x_i)$ for $i=1,\ldots,12$. This constitutes theoretical motivation for working with partially phased data, even if we, in principle, have finite identifiability already from the unphased data.
\end{remark}

\begin{remark}
When $\alpha=2$, which corresponds to the setting studied in~\cite{belyaeva2021identifying}, then computationally we found that for some special choices of $x_1^*,y_1^*,\ldots,x_{12}^*,y_{12}^* \in \mathbb{R}^3$ the rank of the Jacobian matrix in Theorem~\ref{theorem:ambiguous-noiseless-finite-identifiability} is $42$. This is consistent with the fact that Theorem~\ref{theorem:ambiguous-noiseless-finite-identifiability} fails for $\alpha=2$~\cite{belyaeva2021identifying}.
\end{remark}

\section{A new reconstruction method}
\label{sec:reconstruction}
In this section, we outline a new approach to diploid 3D genome reconstruction for partially phased data, based on the theoretical results discussed in subsection \ref{subsec:partially_ambigous}. The method consists of the following main steps:
\begin{enumerate}
    \item Estimation of the unambiguous beads $\{x_i,y_i\}_{i\in U}$ through semidefinite programming (discussed in subsection~\ref{subsec:estimation_of_unambiguous_beads}).
    \item A preliminary estimation of the ambiguous beads using numerical algebraic geometry, based on Corollary~\ref{cor:finite_identifiability_noisy} (discussed in subsection~\ref{subsec:NAG}).
    \item A refinement of this estimation using local optimization (discussed in subsection~\ref{subsec:local_optimization}).
    \item A final clustering step, where we disambiguate between the estimations $(x_i,y_i)$ and $(y_i,x_i)$ for each $i\in A$, based on the assumption that homolog chromosomes are separated in space (discussed in subsection~\ref{subsec:clustering}).
\end{enumerate}
In what follows, we will refer to this method by the acronym SNLC (formed from the initial letters in semidefinite programming, numerical algebraic geometry, local optimization and clustering). 

\subsection{Estimation of the positions of unambiguous beads} \label{subsec:estimation_of_unambiguous_beads}
As discussed in section \ref{subsec:euclidean_distance_geometry}, the unambiguous bead coordinates  $\{x_i,y_i\}_{i\in U}=\{z_i\}_{i\in U\cup (n+U)}$ can be estimated with semidefinite programming. 
More specifically, we use ChromSDE~\cite[Section~2.1]{zhang2013inference} for this part of our reconstruction, which relies on a  specialized solver from \cite{kaifeng2014partial}, to solve an SDP relaxation of the optimization problem 
\begin{equation}
\label{eqn:euclidean_distance_problem}
\min_{\{z_i\}_{i\in U\cup(n+U)}} \sum_{\substack{i,j\in U\cup(n+U)\\c_{ij}^U\neq 0}}\sqrt{c_{ij}^U}\left(\frac{1}{c_{ij}^U}-\|z_i-z_j\|^2\right)^2+\lambda\sum_{\substack{i,j\in U\cup(n+U)\\c_{ij}^U=0}} \|z_i-z_j\|^2
\end{equation}
with $\lambda=0.01$ (cf. \cite[Equation~4]{zhang2013inference}). The terms in the first sum are weighted by the square root for the corresponding contact counts, in order to account for the fact that higher counts can be assumed to be less susceptible to noise. 

\subsection{Preliminary estimation using numerical algebraic geometry}
\label{subsec:NAG}

To estimate the coordinates of the ambiguous beads $\{x_i,y_i\}_{i\in A}$, we will use a method based on numerical equation solving, where we estimate the ambiguous bead pairs one by one.

Let $x,y$ be the unknown coordinates in $\RR^3$ of a pair of ambiguous beads. We pick six unambiguous beads with already estimated coordinates $a,b,c,d,e,f \in \mathbb{R}^3$. For each $t\in\{a,\ldots,f\}$, let $c_{t}\in \RR$ be the corresponding partially ambiguous counts between $t$ and the ambiguous bead pair $(x,y)$. Clearing the denominators in the system~(\ref{rational_eq_noisy}), we obtain a system of polynomial equations
\begin{equation}\label{poly_eq_noisy}
   \|x-t\|^2 + \|y-t\|^2 = c_t\|x-t\|^2 \|y-t\|^2 \text{ for } t=a,b,c,d,e,f.
\end{equation}
By Corollary~\ref{cor:finite_identifiability_noisy}, this system has finitely many complex solutions both in the noiseless and noisy setting, which can be found using homotopy continuation.

We observe that the system \eqref{poly_eq_noisy} generally has 80 complex solutions, and we only expect one pair of solutions $(x,y),(y,x)$ to correspond to an accurate estimation. Naively adding another polynomial arising from a seventh unambiguous bead (as in Conjecture~\ref{conj:unique_identifiability_alpha=2}) does not work; in the noisy setting this over-determined system typically lacks solutions. Instead, we compute an estimation based on the following two heuristic assumptions:
\begin{enumerate}
    \item The most accurate estimation should be \textit{approximately real}, in the sense that the max-norm of the imaginary part is below a certain tolerance (in this work, 0.15 was used for the experiments in both subsections~\ref{subsec:synthetic} and~\ref{subsec:real_dataset}). The choice of this threshold was made based on analysing the imaginary parts of solutions to \eqref{poly_eq_noisy} for various choices of unambiguous beads, see Figure~\ref{fig:norms_of_imaginary_parts}.
    \item The most accurate estimation should be consistent when we change the choice of six unambiguous beads.
\end{enumerate}
Based on these assumptions, we apply the following strategy. We make a number $N\geq 2$, choices of sets of six unambiguous beads, and solve the corresponding $N$ square systems of the form \eqref{poly_eq_noisy}. Since larger contact counts can be expected to have smaller relative noise, we make the choices of beads among the 20 unambiguous beads $t$ that have highest contact count $c_t$ to the ambiguous locus at hand. For each system, we pick out the approximately real solutions, and obtain $N$ sets $\mathcal{S}_1,\ldots,\mathcal{S}_N\subseteq\mathbb{R}^6$ consisting of the real parts of the approximately real solutions. Up to the symmetry $(x,y)\mapsto (y,x)$, we expect these sets to have a unique ``approximately common'' element. 
We therefore compute, by an exhaustive search, the tuple $(w_1,\ldots,w_N)\in\mathcal{S}_1\times\cdots\times\mathcal{S}_N$ that minimizes the sum \[\left\Vert w_1-\frac{w_1+\cdots+w_N}{N}\right\Vert+\cdots+\left\Vert w_N-\frac{w_1+\cdots+w_N}{N}\right\Vert,\] 
and use $\frac{w_1+\cdots+w_N}{N}$ as our estimation of $(x,y)$. For the computations presented in section~\ref{sec:experiments}, we use $N=5$. 

To solve the systems, we use the Julia package \verb$HomotopyContinuation.jl$ \cite{homotopycontinuation}, and follow the two-phase procedure described in \cite[Section~7.2]{sommese2005numerical}. For the first phase, we solve \eqref{poly_eq_noisy} with randomly chosen parameters $a^*,\ldots,f^*\in\CC^3$ and $c_{a^*},\ldots,c_{f^*}\in\CC$, using a polyhedral start system \cite{Birkett1995polyhedral}. We trace 1280 paths in this first phase, since the Newton polytopes of the polynomials appearing in the system \eqref{poly_eq_noisy} all contain the origin, and have a mixed volume of 1280, which makes 1280 an upper bound on the number of complex solutions by \cite[Theorem~2.4]{Li1996bkk}. For the second phase, we use a straight-line homotopy in parameter space from the randomly chosen parameters $a^*,\ldots,f^*\in\CC^3$ and $c_{a^*},\ldots,c_{f^*}\in\CC$, to the values $a,\ldots,f$ and  $c_{a},\ldots,c_{f}\in\CC$ at hand. We observe that we generally find 80 complex solutions in the first phase, which means 40 orbits with respect to the symmetry $(x,y)\mapsto (y,x)$. By the discussion in \cite[Section~7.6]{sommese2005numerical}, it is enough to only trace one path per orbit, so in the end, we only trace 40 paths in the second phase.

\begin{remark}
If the noise levels are sufficiently high, there could be choices of six unambiguous beads for which the system lacks approximately-real solutions. If this situation is encountered, we  try to redraw the six unambiguous beads until we find an approximately-real solution. If this does not succeed within a certain number of attempts (100 in the experiments conducted for this paper), we use the average of the  closest neighboring unambiguous beads instead.
\end{remark}

\subsection{Local optimization} \label{subsec:local_optimization}

A disadvantage of the numerical algebraic geometry based estimation discussed in the previous subsection is that it only takes into account ``local'' information about the interactions for one ambiguous locus at a time, which might make it more sensitive to noise. In our proposed method, we therefore refine this preliminary estimation of $\{x_i,y_i\}_{i\in A}$ further in a local optimization step that takes into account the ``global''  information of all available data.

The idea is to estimate $\{x_i,y_i\}_{i\in A}$ by solving the optimization problem
\begin{equation}
\label{eq:local_optimization_problem}
\min_{\{x_i,y_i\}_{i\in A}}\:\:\hspace{-2pt}{\sum_{i\in U,j\in A}\hspace{-5pt}\left( \left(c^P_{i,j} \hspace{-1pt} - \hspace{-1pt} \tfrac{1}{\|x_i-x_j\|^2}\hspace{-1pt}-\hspace{-1pt}  \tfrac{1}{\|x_i-y_j\|^2}\right)^2 \hspace{-7pt}+\hspace{-3pt}\left(c^P_{i+n,j} \hspace{-1pt} - \hspace{-1pt}  \tfrac{1}{\|y_i-x_j\|^2} \hspace{-2pt} - \hspace{-2pt}   \tfrac{1}{\|y_i-y_j\|^2}\right)^2\right)}
\end{equation}
while keeping the estimates of $\{x_i,y_i\}_{i\in U}$ from the ChromSDE step fixed. We use the quasi-Newton method for unconstrained optimization implemented in the Matlab Optimization Toolbox for this step. The already estimated coordinates of $\{x_i,y_i\}_{i\in A}$ from the numerical algebraic geometry step are used for the initialization.

\subsection{Clustering to break symmetry} \label{subsec:clustering}

Our objective function remains invariant if we exchange $x_i$ and $y_i$ for any $i\in A$.
We can break symmetry by relying on the empirical observation that homologous chromosomes typically are spatially separated in different so-called compartments of the nucleus~\cite{eagen2018principles}.
Let $(\bar{x}_i,\bar{y}_i)_{i=1}^n$ denote the estimates from the previous steps.
Our final estimations will be obtained by solving the minimization problem 
\begin{align}
  \label{eqn:clustering_minimization}
  \min_{\{x_i,y_i\}_{i\in A}}\; \hspace{-3pt} \sum_{i=1}^{n-1} \;g_{i,i+1}(x,y),
  \text{ with } \,\,
  g_{i,i+1}(x,y):= \left(\|x_i - x_{i+1}\|^2 + \|y_i - y_{i+1}\|^2\right), 
\end{align}
where $(x_i,y_i)=(\bar{x}_i,\bar{y}_i)$ for $i\in U$ are fixed,
and $(x_i,y_i)\in\{(\bar{x}_i,\bar{y}_i),(\bar{y}_i,\bar{x}_i)\}$ for $i\in A$ are the optimization variables.
The optimal solution can be computed efficiently, as explained next.

We first decompose the problem into contiguous chunks of ambiguous beads.
Let $(i_1,\dots,i_{L}) := U$ be the indices of the unambiguous beads
and let $i_0 := 1$, $i_{L+1} := n$.
The optimization problem can be phrased as
\begin{align}
  \label{eqn:clustering_minimization2}
  \min_{\{x_i,y_i\}_{i\in A}}\;
  \sum_{\ell=0}^{L} G_\ell(x,y),
  \quad\text{ with }\quad
  G_\ell(x,y) := \sum_{i=i_\ell}^{i_{\ell+1}-1} \;g_{i,i+1}(x,y)
\end{align}
where there is one summand $G_\ell(x,y)$ for each contiguous chunk of ambiguous beads.
Since the summands $G_\ell(x,y)$ do not share any ambiguous bead,
we can minimize them independently.

We proceed to describe the optimal solution of the problem.
Let
\begin{align*}
  s_i = \begin{cases}
    1, &\text{ if }(x_i,y_i) = (\bar x_i, \bar y_i)\\
    -1, &\text{ if }(x_i,y_i) = (\bar y_i,\bar x_i)
  \end{cases},
  \qquad
  w_{i,i+1} = (\bar{x}_i-\bar{y}_i)^T(\bar{x}_{i+1}-\bar{y}_{i+1}).
\end{align*}
The variable $s_i$ indicates whether we keep using $(\bar x_i, \bar y_i)$ or we reverse it.
Note that $s_i = 1$ for $i \in U$.
The next lemma gives the optimal assignment of $s_i$ for $i \in A$.
This assignment is constructed by using inner products~$w_{i,i+1}$.

\begin{lemma}
  The optimal solution of \eqref{eqn:clustering_minimization} can be constructed as follows:
  \begin{enumerate}
    \item For the last chunk ($\ell = L$) we have
      \begin{align*}
        s_{i_{\ell}}^* = 1, \qquad\quad
        s_{i+1}^* = \sgn(w_{i,i+1})s_{i}^*
        \quad\text{ for } i=i_{\ell}, i_{\ell}{+}1, \dots, i_{\ell+1}{-}1
      \end{align*}
      where $\sgn(\cdot)$ is the sign function and  $\sgn(0)$ can be either $1$ or $-1$.
    \item For the first chunk ($\ell=0$) we have
      \begin{align*}
        s_{i_{\ell+1}}^* = 1, \qquad\quad
        s_{i}^* = \sgn(w_{i,i+1})s_{i+1}^*
        \quad\text{ for } i=i_{\ell+1}{-}1, i_{\ell+1}{-}2, \dots, i_\ell
      \end{align*}
    \item For any other chunk,
      let $k$ be the index of the smallest absolute value $|w_{k,k+1}|$, among $i_{\ell}\leq k \leq {i_{\ell+1}-1}$.
      The solution is
      \begin{align*}
        s_{i_{\ell}}^* &= 1, \qquad\quad
        s_{i+1}^* = \sgn(w_{i,i+1})s_{i}^*
        \quad\text{ for } i=i_{\ell}, i_{\ell}{+}1, \dots, k{-}1
        \\
        s_{i_{\ell+1}}^* &= 1, \qquad\quad
        s_{i}^* = \sgn(w_{i,i+1})s_{i+1}^*
        \quad\text{ for } i=i_{\ell+1}{-}1, i_{\ell+1}{-}2, \dots, k{+}1
      \end{align*}
  \end{enumerate}
\end{lemma}

\begin{proof}
  Denoting $\bar u_i := \tfrac{1}{2}(\bar{x}_i + \bar{y}_i)$, $\bar v_i := \tfrac{1}{2}(\bar{x}_i - \bar{y}_i)$,
  then $ x_i = u_i + s_i v_i$, $ y_i = u_i - s_i v_i $.
  Note that
  
  \begin{align*}
    \|\bar{x}_i\|^2 + \|\bar{y}_i\|^2 + \|\bar{x}_{i+1}\|^2 &+ \|\bar{y}_{i+1}\|^2
    - g_{i,i+1}(x,y)
    = 2 (x_i^T x_{i+1} + y_i^T y_{i+1})\,
    \\
    &{\hspace{-40pt}}=2(\bar u_i + s_i \bar v_i)^T (\bar u_{i+1} + s_{i+1} \bar v_{i+1})
    + 2(\bar u_i - s_i \bar v_i)^T (\bar u_{i+1} - s_{i+1} \bar v_{i+1})\\
     &{\hspace{-40pt}}= 4 (\bar u_i^T \bar u_{i+1}) + 4 (\bar v_i^T \bar v_{i+1}) s_i s_{i+1}
    \\
     &{\hspace{-40pt}}= 4 (\bar u_i^T \bar u_{i+1}) + w_{i,i+1} s_i s_{i+1}
  \end{align*}
  
  Since $\bar x_i, \bar y_i, \bar u_i, \bar v_i$ are constants,
  minimizing $g_{i,i+1}(x,y)$ is equivalent to maximizing $w_{i,i+1} s_i s_{i+1}$.
  Then for each chunk we have to solve the optimization problem
  \begin{align}\label{eqn:clustering_maximization}
    \max_{s_i\in\{1,-1\}} \;\;\sum_{i=i_{\ell}}^{i_{\ell+1}-1} w_{i,i+1} s_i s_{i+1}\,,
  \end{align}

  The formulas from the first and last chunk are such that $w_{i,i+1} s_i^* s_{i+1}^* \geq 0$ for all~$i$.
  This is possible because in these cases only one of the endpoints has a fixed value,
  and the remaining values are computed recursively starting from such a fixed point.
  Since all summands are nonnegative, the sum in \eqref{eqn:clustering_maximization} is maximized.

  For the inner chunks, the two endpoints are fixed, so it may not be possible to have that $w_{i,i+1} s_i^* s_{i+1}^* \geq 0$ for all indices.
  In an optimal assignment we should pick at most one term to be negative,
  and such a term (if it exists) should be the one with the smallest absolute value $|w_{i,i+1}|$.
  This leads to the formula from the lemma.
\end{proof}

\section{Experiments}
\label{sec:experiments}

In this section, we apply the SNLC scheme described in section~\ref{sec:reconstruction} to synthetic and real datasets, and compare its performance with the preexisting software packages ASHIC~\cite{ye2020ashic} and PASTIS~\cite{cauer2019inferring}. We chose these two reconstruction methods for comparison because they are best suited for our setting. Also Belyaeva et al.\ \cite{belyaeva2021identifying} and Tan et al.\ \cite{tan2018three} have reconstruction methods for diploid organisms, but the former method requires higher-order contact information and the latter method is targeted for single cell data.

All SNLC experiments are done using Julia~1.6.1, with ChromSDE being run in Matlab 2021a, and the Julia package \texttt{MATLAB.jl} (v0.8.3) acting as interface between Julia and Matlab. The numerical algebraic geometry part of the estimation procedure is done with \texttt{HomotopyContinuation.jl} (v2.5.5) \cite{homotopycontinuation}. The PASTIS experiments are run in Python 3.8.10, and the ASHIC experiments in Python 3.10.5. 

For the PASTIS computations, we fix $\alpha=-2$ to ensure compatibility with the modelling assumptions made in this paper. We run PASTIS without filtering, in order to make it possible to compare RMSD values. Since PASTIS only takes integer inputs, we multiply the theoretical contact counts calculated by \eqref{eq:full_system} by a factor $10^5$ and round them to the nearest integer. Following the approach taken in \cite{cauer2019inferring}, we use a coarse grid search to find the optimal coefficients for the homolog separating constraint and bead connectivity constraints. Specifically, we fix a structure simulated with the same method as used in the experiments, and compute the RMSD values for all $\lambda_1,\lambda_2\in\{1,10^1,10^2,\ldots,10^{12}\}$. In this way, we find that $\lambda_1=10^{11}$ and $\lambda_2=10^{12}$ give optimal results.

For the ASHIC computations, we use the ASHIC-ZIPM method, which has the lowest distance error rate among the ASHIC's models according to \cite[Figure 2]{ye2020ashic} and models the contact counts as a zero-inflated Poisson distribution (ZIP) to account for the sparsity of the Hi-C matrix.
We run ASHIC without filtering out any loci and with the setting \verb|aggregate| to ensure that the coordinates of all beads are estimated.

\subsection{Synthetic data}
\label{subsec:synthetic}
We conduct a number of experiments where we simulate a single chromosome pair (referred to as $X$ and $Y$ in figures) through Brownian motion with fixed step length, compute unambiguous, partially ambiguous and ambiguous contact counts according to \eqref{eq:full_system}, add noise, and then try to recover the structure of the chromosomes through the SNLC scheme described in section~\ref{sec:reconstruction}. Following \cite{belyaeva2021identifying}, we model noise by multiplying each entry of $C^U$, $C^P$ and $C^A$ by a factor $1+\delta$, where $\delta$ is sampled uniformly from the interval $(-\varepsilon,\varepsilon)$ for some chosen noise level $\varepsilon\in [0,1]$.

As a measure of the quality of the reconstruction, we use the minimal root-mean square distance (RMSD) between, on the one hand, the true coordinates $(x_i^*,y_i^*)_{i=1}^n$, and, on the other hand, the estimated coordinates $(x_i,y_i)_{i=1}^n$ after rigid transformations and scaling, i.e., we find the minimum
\[\min_{\substack{R\in \mathrm{O}(3)\\s>0,\: b\in\mathbb{R}^3}}\sqrt{\frac{1}{2n} \sum_{i=1}^n \Big(\Vert (sR x_i+b)-x_i^*\Vert^2+\Vert (sR y_i+b)-y_i^*\Vert^2\Big)}.\]
This can be seen as a version of the classical Procrustes problem solved in \cite{schonemann1966generalized}, which is implemented in Matlab as the function \verb$procrustes$.

Specific examples of reconstructions of the Brownian motion and helix-shaped chromosomes obtained with SNLC at varying noise levels and $50\%$ of ambiguous beads are shown in Figure~\ref{fig:examples_of_reconstructions}. For low noise levels the reconstructions by SNLC and the original structure highly overlap. For higher noise levels the general region occupied by the reconstructions overlaps with the original structure, while the local features become less aligned.
Analogous reconstructions obtained with SNLC without the local optimization step 
are shown in Figure~\ref{fig:reconstructions_without_local} in Appendix. 

A comparison of how the quality of the reconstruction depends on the noise level and proportion of ambiguous beads for SNLC, ASHIC and PASTIS is done in Figure~\ref{fig:comparison_rmsd_vs_noise}. We measure the RMSD value between the reconstructed and original 3D structure for different noise levels over 20 runs. The RMSD values obtained by SNLC are consistently lower than the ones obtained by ASHIC and PASTIS. The difference is specially large for low to medium noise levels. 
While our method outperforms ASHIC and PASTIS  in the setting considered in this paper, it is worth mentioning that ASHIC and PASTIS work also in a more general setting, where there might be contacts of all three types (ambiguous, partially ambiguous and unambiguous) between every pair of loci.

\begin{figure}[ht]
    \begin{subfigure}{0.3\textwidth}
    \includegraphics[width=\textwidth]{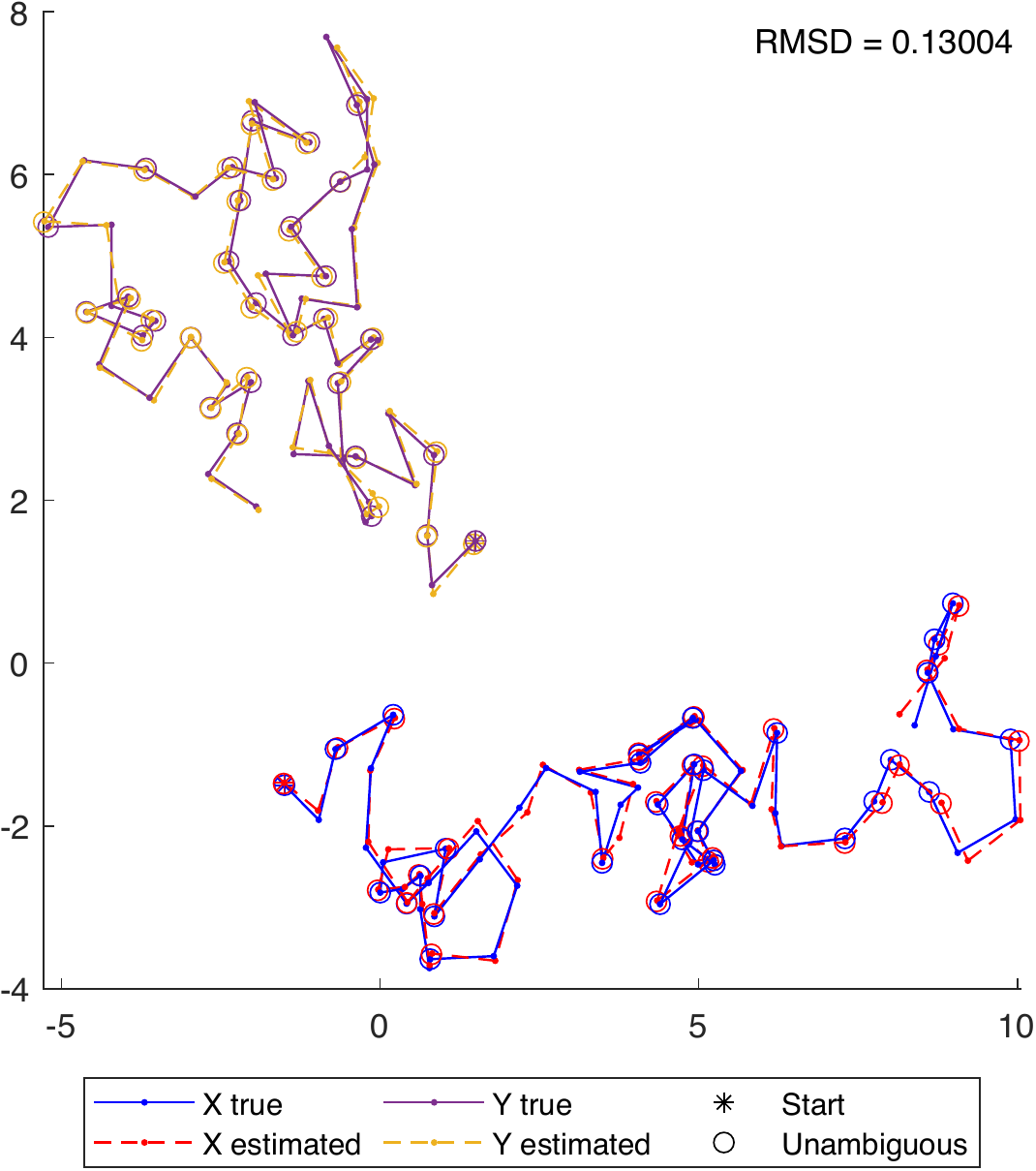}
    \caption{$\varepsilon=0.10$}
    \end{subfigure}
    ~
    \begin{subfigure}{0.3\textwidth}
    \includegraphics[width=\textwidth]{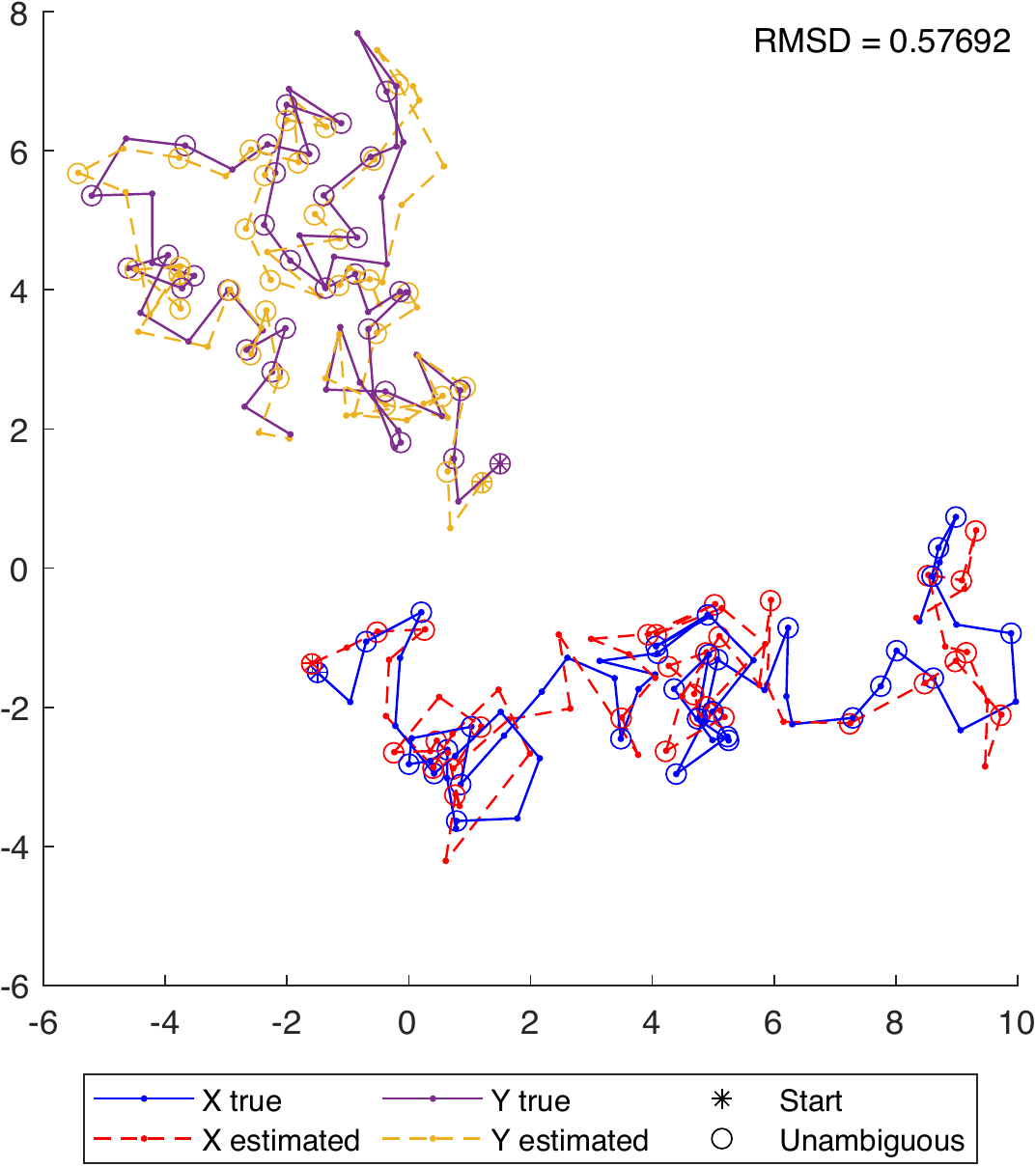}
    \caption{$\varepsilon=0.50$}
    \end{subfigure}
    ~
    \begin{subfigure}{0.3\textwidth}
    \includegraphics[width=\textwidth]{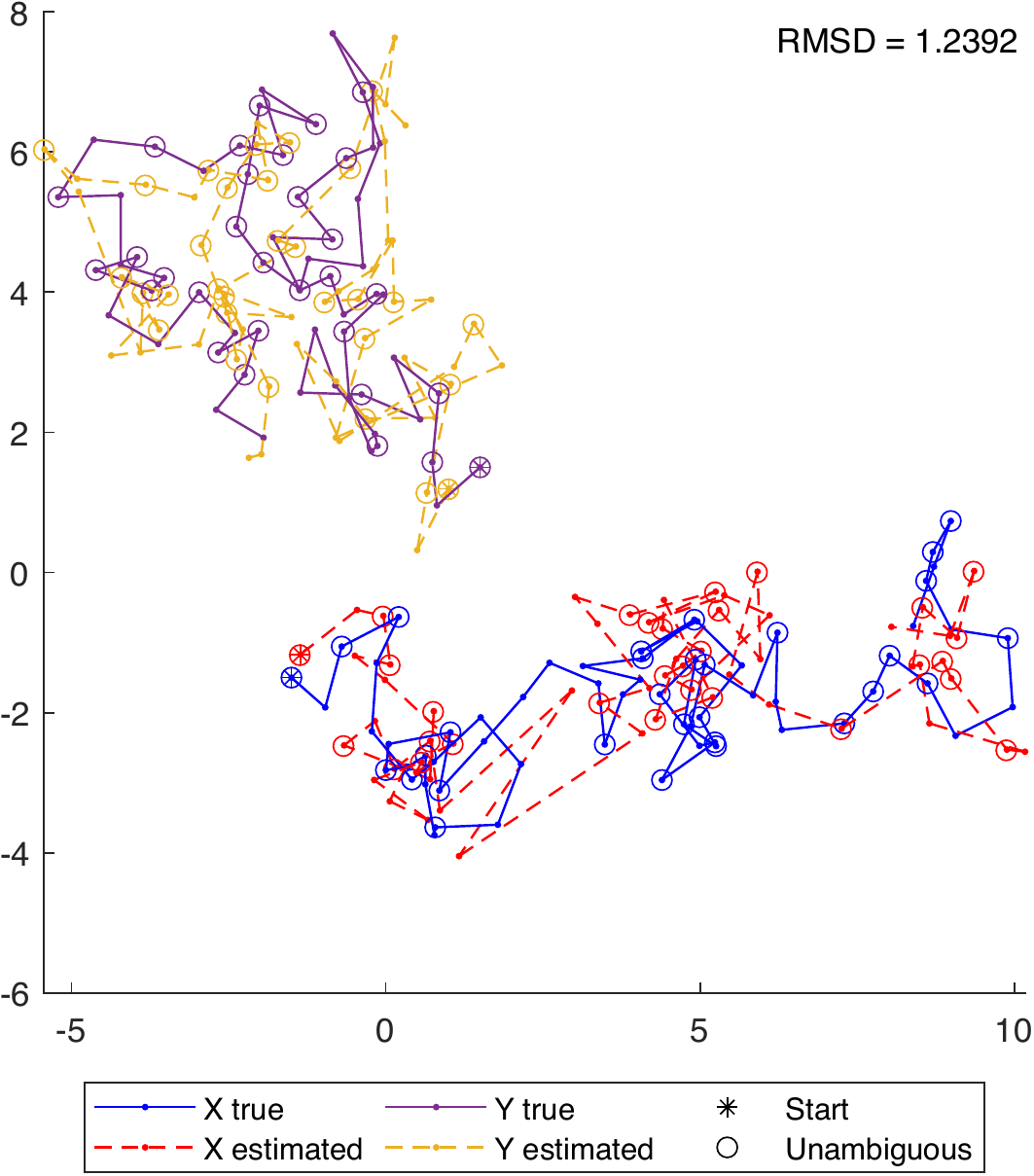}
    \caption{$\varepsilon=0.90$}
    \end{subfigure}\\[2em]
    \begin{subfigure}{0.3\textwidth}
    \includegraphics[width=\textwidth]{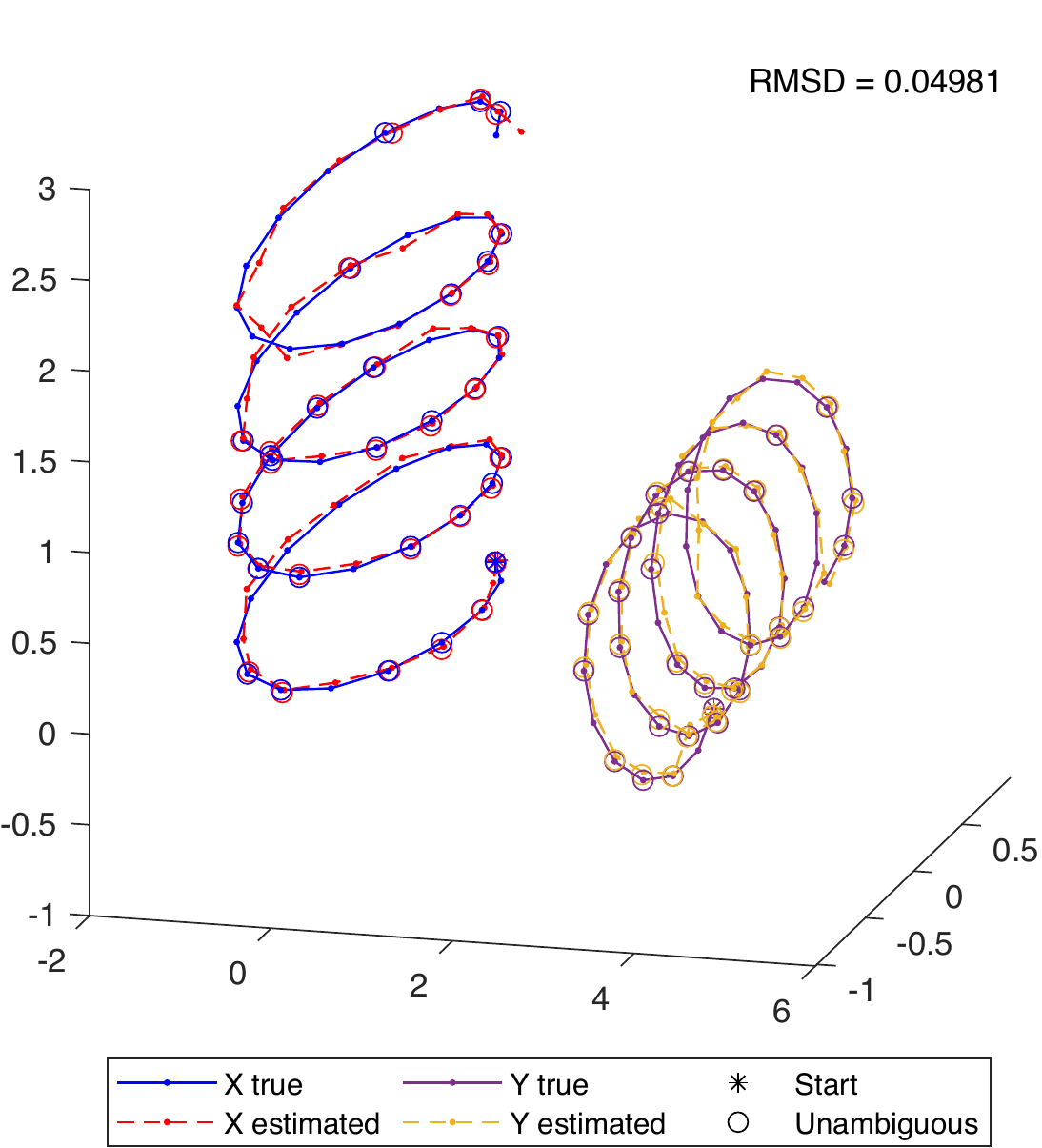}
    \caption{$\varepsilon=0.10$}
    \end{subfigure}
    ~
    \begin{subfigure}{0.3\textwidth}
    \includegraphics[width=\textwidth]{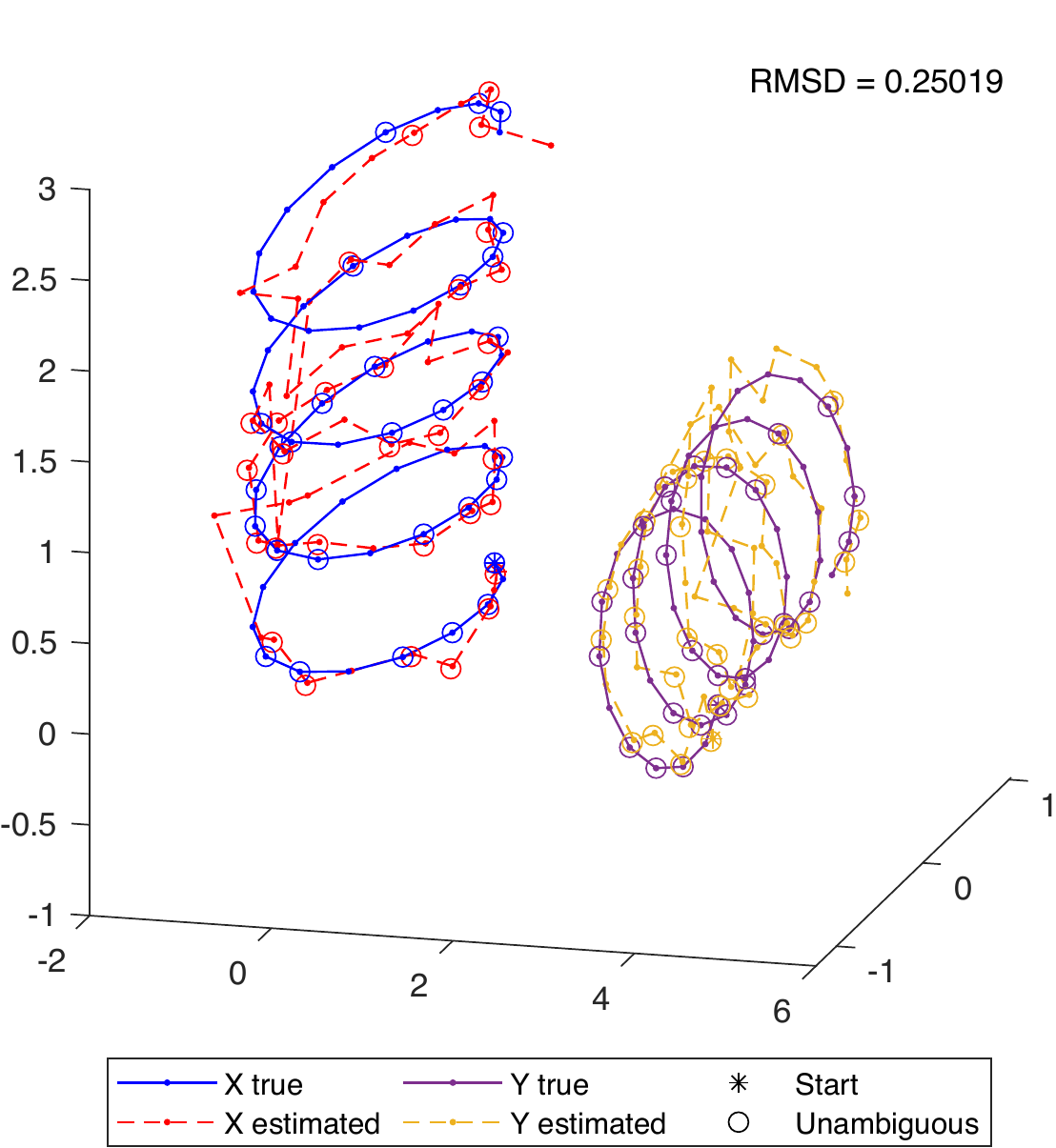}
    \caption{$\varepsilon=0.50$}
    \end{subfigure}
    ~
    \begin{subfigure}{0.3\textwidth}
    \includegraphics[width=\textwidth]{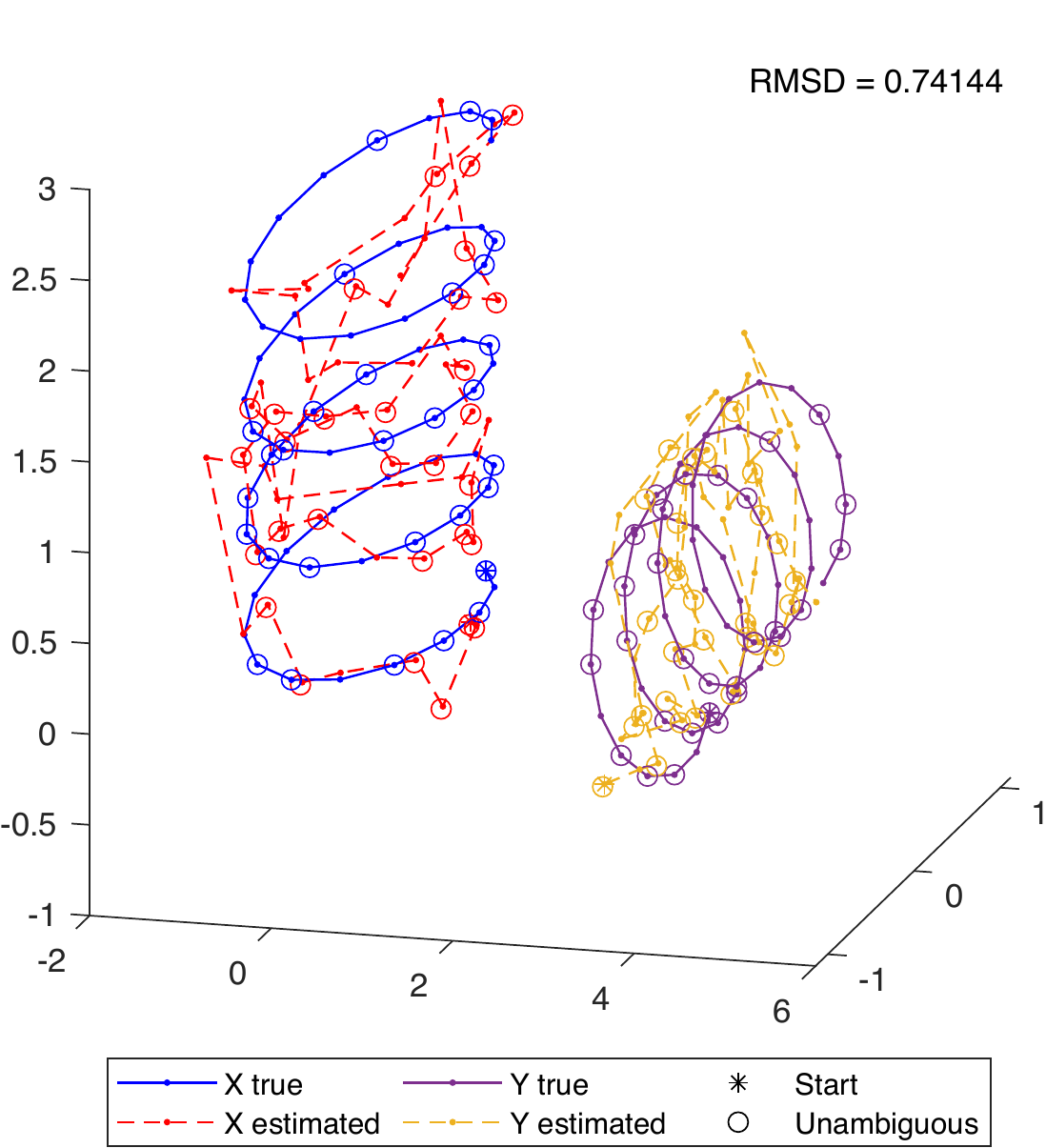}
    \caption{$\varepsilon=0.90$}
    \end{subfigure}
    \centering
    \caption{Examples of reconstructions for varying noise levels, for a chromosome pair with 60 loci, out of which $50\%$ are ambiguous. Subfigures (a)--(c) show chromosomes simulated with Brownian motion (projected onto the $xy$-plane), whereas figure (d)--(e) show helix-shaped chromosomes.}
    \label{fig:examples_of_reconstructions}
\end{figure}

\begin{figure}[ht]
    \begin{subfigure}{0.4\textwidth}
    \includegraphics[width=\textwidth]{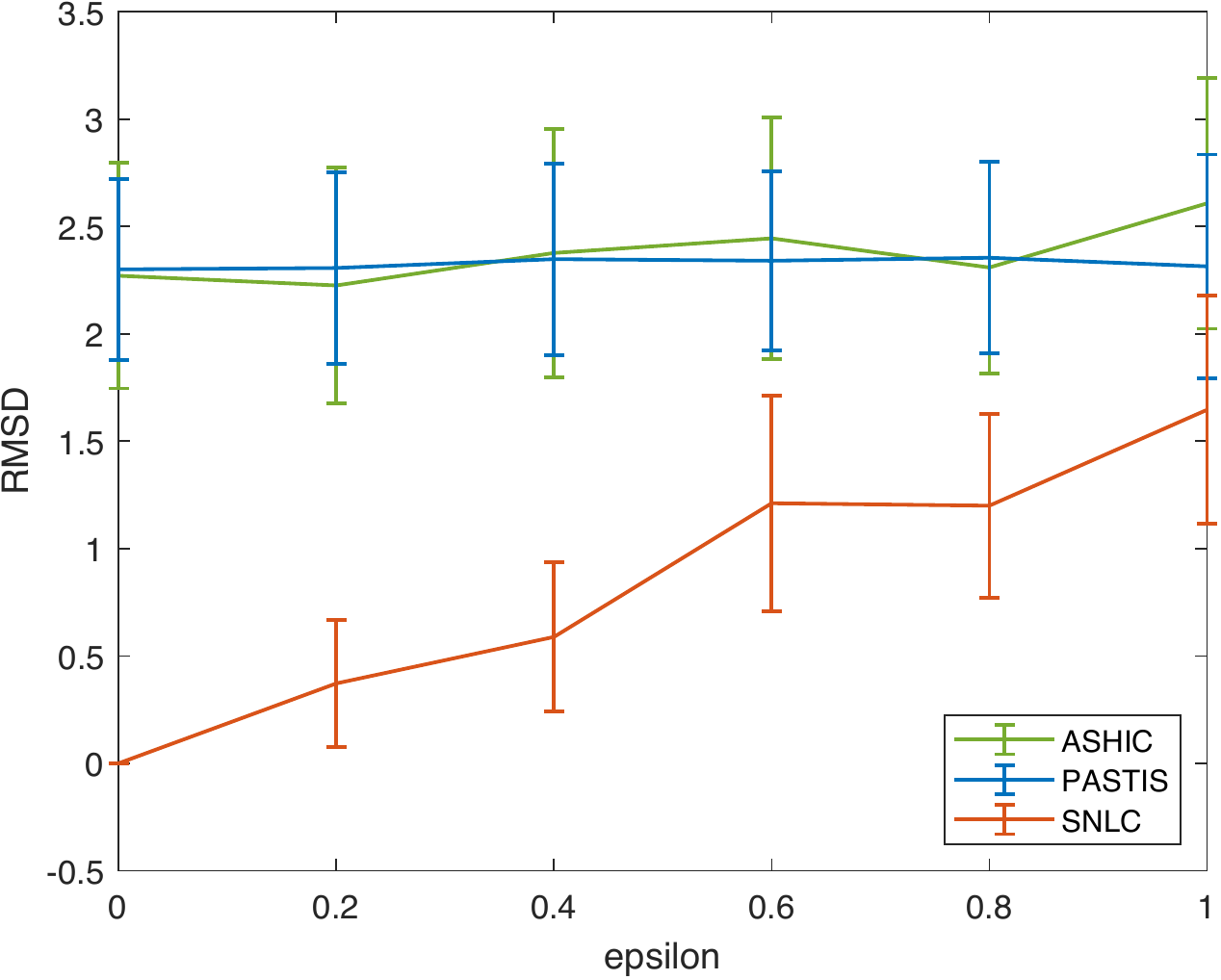} 
    \caption{25\% ambiguous loci}
    \end{subfigure}
    \hspace{2em}
    \begin{subfigure}{0.4\textwidth}
    \includegraphics[width=\textwidth]{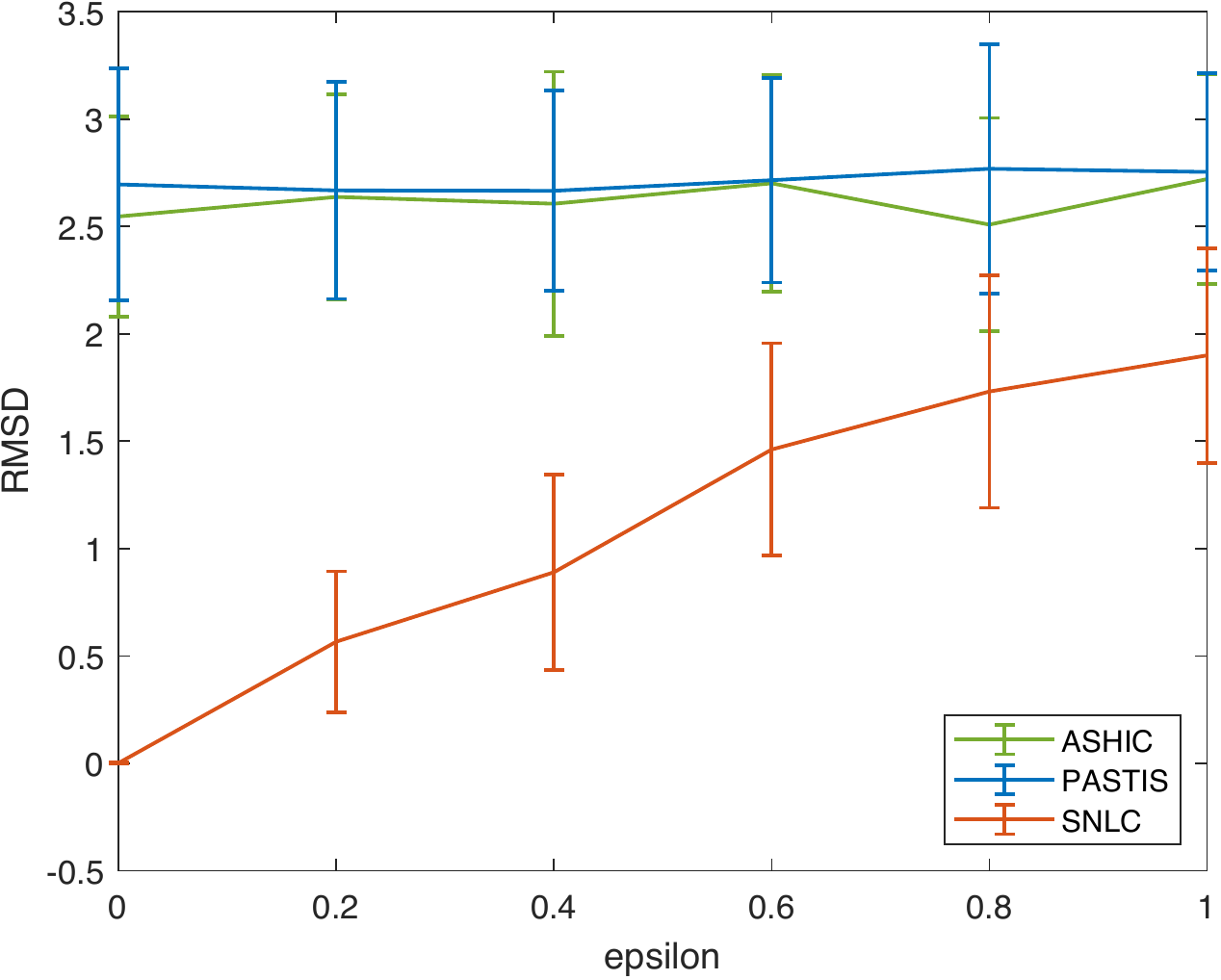} 
    \caption{50\% ambiguous loci}
    \end{subfigure}
    \centering
    \caption{Comparison between our reconstruction method, ASHIC and PASTIS. The values are the average over 20 runs, with the error bars showing the standard deviation. All experiments took place with 60 loci, with varying levels of noise, as well as varying numbers of ambiguous loci, uniformly randomly distributed over the chromosomes.}
    \label{fig:comparison_rmsd_vs_noise}
\end{figure}

\subsection{Experimentally obtained data}
\label{subsec:real_dataset}
We compute SNLC reconstructions based on the real dataset explored in \cite{cauer2019inferring}, which is obtained from Hi-C experiments on the X chromosomes in the Patski (BL6xSpretus) cell line. The data has been recorded at a resolution of 500 kb, which corresponds to 343 bead pairs in our model.

For some of these pairs, no or only very low contact counts have been recorded. Since such low contact counts are susceptible to high uncertainty and can be assumed to be a consequence of experimental errors, we exclude the $47$ loci with the lowest total contact counts from the analysis. To select the cutoff, the loci are sorted according to the total contact counts (see Figure~\ref{fig:real_data_preprocessing}~(a) in Appendix), and the ratios between the total contact counts for consecutive loci are computed. A peak for these ratios is observed at the 47th contact count, as shown in Figure~\ref{fig:real_data_preprocessing}~(b) in Appendix. 
After applying this filter, we obtain a dataset with $296$ loci. 
Out of these, we consider as ambiguous all loci $i$ for which less than $40\%$ of the total contact count comes from contacts where $x_i$ and $y_i$ were not distinguishable. These proportions for all loci are shown in Figure~\ref{fig:real_data_preprocessing}~(c) in Appendix. For the Patski dataset, we obtain 46 ambiguous loci and 250 unambiguous loci in this way.

In the Patski dataset, a locus can simultaneously participate in unambiguous, partially ambiguous and ambiguous contacts. To obtain the setting of our paper where loci are partitioned into unambiguous or ambiguous, we reassign the contacts according to whether a locus is unambiguous or ambiguous. Our reassignment method is motivated by the assignment of haplotype to unphased Hi-C reads in~\cite{lindsly2021functional}. The exact formulas are given in Appendix.

The reconstruction obtained via SNLC can be found in Figure~\ref{fig:real_reconstructions} (a).   The logarithmic heatmaps for contact count matrices for original data and the SNLC reconstruction are shown in Figure~\ref{fig:reconstructed_contacts}.

It was discovered in~\cite{deng2015bipartite} that the inactive homolog 
in the Patski X chromosome pair has a bipartite structure, consisting of two superdomains with frequent intra-chromosome contacts within the superdomains and a boundary region between the two superdomains. The active homolog 
does not exhibit the same behaviour. The boundary region on the inactive X chromosome is centered at 72.8-72.9 MB~\cite{deng2015bipartite} which at the 500 kB resolution corresponds to the bead 146~\cite{cauer2019inferring}.  We show in Figure~\ref{fig:real_reconstructions} (b) that the two chromosomes reconstructed using SNLC exhibit this structure by computing the bipartite index for the respective homologs as in~\cite{cauer2019inferring,deng2015bipartite}. We recall that, in the setting of a single chromosome with beads $z_1,\ldots,z_n\in\mathbb{R}^3$, the bipartite index is defined as the ratio of intra-superdomain to inter-superdomain contacts in the reconstruction:
\begin{equation*}
  BI(h) =
  \frac{\tfrac{1}{h^{2}}\sum_{i=1}^{h} \sum_{j=1}^{h} \frac{1}{\Vert z_i-z_j\Vert^2}+\tfrac{1}{(n-h)^{2}}\sum_{i=h+1}^{n} \sum_{j=h+1}^{n}  \frac{1}{\Vert z_i-z_j\Vert^2}}
  {\tfrac{2}{h(n-h)}\sum_{i=1}^{h} \sum_{j=h+1}^{n}  \frac{1}{\Vert z_i-z_j\Vert^2}}.
\end{equation*}

\begin{figure}[ht]
    \centering
    \begin{subfigure}{0.4\textwidth}
    \includegraphics[width=\textwidth]{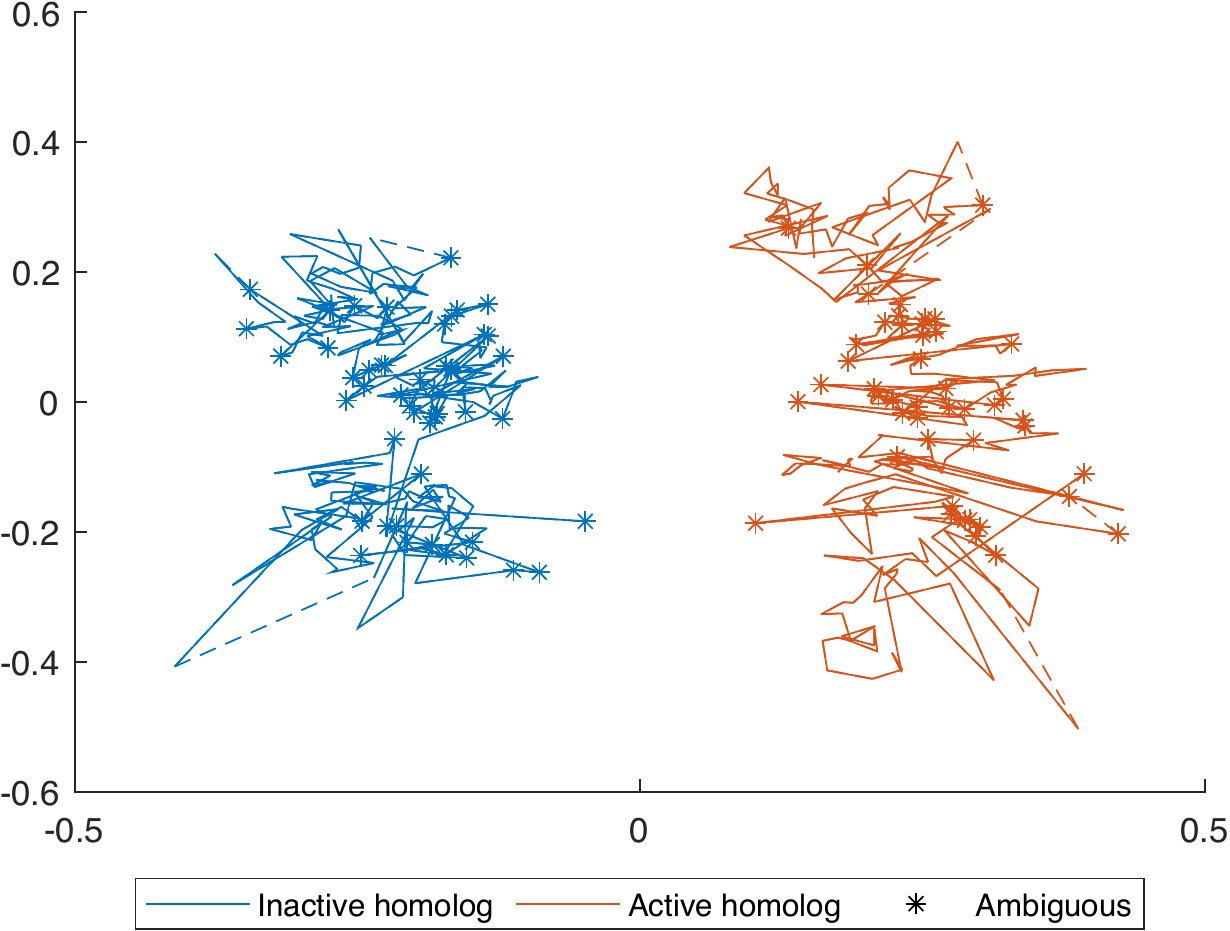}
    \caption{}
    \end{subfigure}
    \hspace{2em}
    \begin{subfigure}{0.4\textwidth}
    \vspace{0.5em}
    \includegraphics[width=\textwidth]{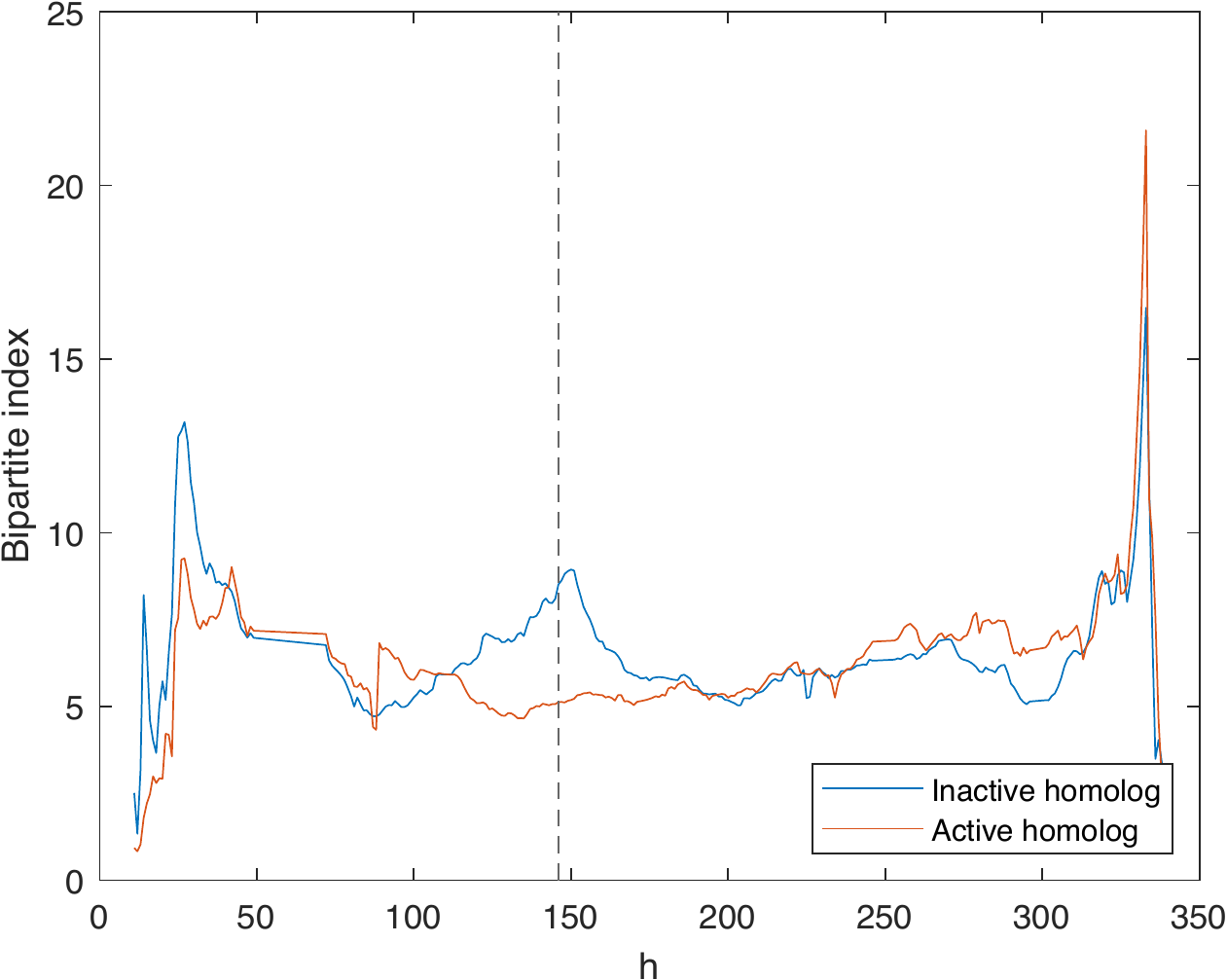}
    \caption{}
    \end{subfigure}
    
    \centering
    \caption{(a) Reconstruction from a real dataset using our reconstruction method. A dashed line between two beads is used to indicate that there is one or more beads between them, for which we have not given an estimation (due to low contact counts). (b) Bipartite index for the reconstructed chromosomes. The dashed vertical line indicates the known hinge point at locus 146.
}
    \label{fig:real_reconstructions}
\end{figure}

\section{Discussion} \label{discussion}

In this article we study the finite identifiability of 3D genome reconstruction from contact counts under the model where the distances $d_{i,j}$ and contact counts $c_{i,j}$ between two beads $i$ and $j$ follow the power law dependency $c_{i,j} = d_{i,j}^{\alpha}$ for a conversion factor $\alpha < 0$.
We show that if at least six beads are unambiguous, then the locations of the rest of the beads can be finitely identified from partially ambiguous contact counts for rational $\alpha$ satisfying $\alpha <0$ or $\alpha > 2$. In the fully ambiguous setting, we prove finite identifiability for $\alpha=-2$, given ambiguous contact counts for at least 12 pairs of beads. From~\cite{belyaeva2021identifying} it is known that finite identifiability does not hold in the fully ambiguous setting for $\alpha=2$. It is an open question whether finite identifiability of 3D genome reconstruction holds for other $\alpha \in \RR \backslash \{-2,2\}$ in the fully ambiguous setting and for rational $\alpha \in (0,2]$ in the partially ambiguous setting. We conjecture that in the partially ambiguous setting seven unambiguous loci guarantee unique identifiability of the 3D reconstruction for rational $\alpha <0$  or $\alpha > 2$. When $\alpha=-2$, then one approach to studying the unique identifiability might be via the degree of a parametrized family of algebraic varieties.

After establishing the identifiability, we suggest a reconstruction method for the partially ambiguous setting with $\alpha=-2$ that combines semidefinite programming, homotopy continuation in numerical algebraic geometry, local optimization and clustering. To speed up the homotopy continuation based part, we observe that the parametrized system of polynomial equations corresponding to six unambiguous beads has 40 pairs of complex solutions and we trace one path for each orbit. It is an open question to prove that for sufficiently general parameters the system has 40 pairs of complex solution. This question again reduces to studying the degree of a family of algebraic varieties. While our goal is to highlight the potential of our method, one could further regularize its output and use interpolation for the beads that are far away from the neighboring beads. A future research direction is to explore whether numerical algebraic geometry or semidefinite programming based methods can be proposed also for other conversion factors $\alpha < 0$.

\subsection*{Acknowledgments}
We thank Anastasiya Belyaeva, Gesine Cauer, AmirHossein Sadegemanesh, Luca Sodomaco, and Caroline Uhler for very helpful discussions and answers to our questions. 

\subsection*{Declarations}
Oskar Henriksson and Kaie Kubjas were partially supported by the Academy of Finland Grant No.~323416. Oskar Henriksson was also partially funded by the Novo Nordisk project with
grant reference number NNF20OC0065582.

The Patski dataset analyzed in subsection~\ref{subsec:real_dataset} comes from the third-party repository \url{https://noble.gs.washington.edu/proj/diploid-pastis/}, and is based on the dataset GSE68992 from the Gene Expression Omnibus, available at \url{https://www.ncbi.nlm.nih.gov/geo/query/acc.cgi?acc=GSE68992}.

The code used for generating the synthetic data discussed in subsection~\ref{subsec:synthetic} is available in the GitHub repository 
\url{https://github.com/kaiekubjas/3D-genome-reconstruction-from-partially-phased-HiC-data}. This repository also contains the code used for the computations referred to in the discussion preceeding Conjecture~\ref{conj:unique_identifiability_alpha=2}, and in the proof of Theorem~\ref{theorem:ambiguous-noiseless-finite-identifiability}.

This version of the article has been accepted for publication, after peer review 
but is not the Version of Record and does not reflect post-acceptance improvements, or any
corrections. The Version of Record is available online at: \url{http://dx.doi.org/10.1007/s11538-024-01263-7}.

\newlength{\bibitemsep}\setlength{\bibitemsep}{.2\baselineskip plus .05\baselineskip minus .05\baselineskip}
\newlength{\bibparskip}\setlength{\bibparskip}{0pt}
\let\oldthebibliography\thebibliography
\renewcommand\thebibliography[1]{
  \oldthebibliography{#1}
  \setlength{\parskip}{2\bibitemsep}
  \setlength{\itemsep}{\bibparskip}
}

\bibliographystyle{plain}

\bigskip \medskip
\noindent {\bf Authors' addresses:}

\noindent 
Diego Cifuentes, Georgia Institute of Technology\hfill{\tt diego.cifuentes@isye.gatech.edu}\\
Jan Draisma, University of Bern \hfill{\tt jan.draisma@math.unibe.ch}\\
Oskar Henriksson, University of Copenhagen \hfill {\tt oskar.henriksson@math.ku.dk}\\
Annachiara Korchmaros, University of Leipzig\hfill {\tt annachiara@bioinf.uni-leipzig.de}\\
Kaie Kubjas, Aalto University \hfill {\tt kaie.kubjas@aalto.fi}

\newpage

\appendix
\section*{Appendix} \label{supplementary}
In this part of the paper, we include additional details and figures for the experiments in section~\ref{sec:experiments}. 

Figure~\ref{fig:reconstructions_without_local} shows reconstructions of the same chromosomes as displayed in Figure~\ref{fig:examples_of_reconstructions} but without the local optimization step, indicating that semidefinite programming, numerical algebraic geometry and clustering alone can recover the main features of the 3D structure. 

\begin{figure}[h]
    \begin{subfigure}{0.3\textwidth}
    \includegraphics[width=\textwidth]{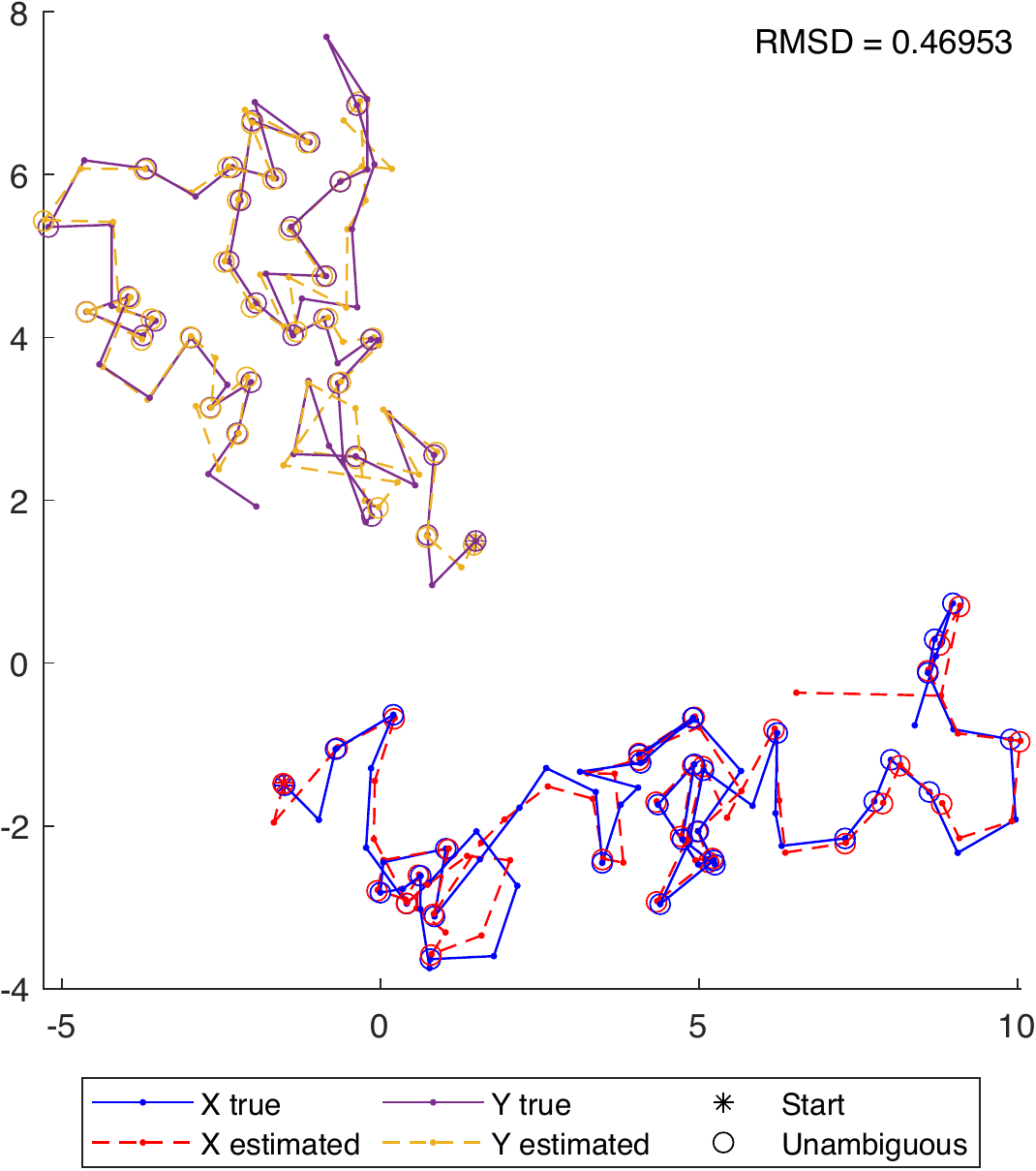}
    \caption{$\varepsilon=0.10$}
    \end{subfigure}
    ~
    \begin{subfigure}{0.3\textwidth}
    \includegraphics[width=\textwidth]{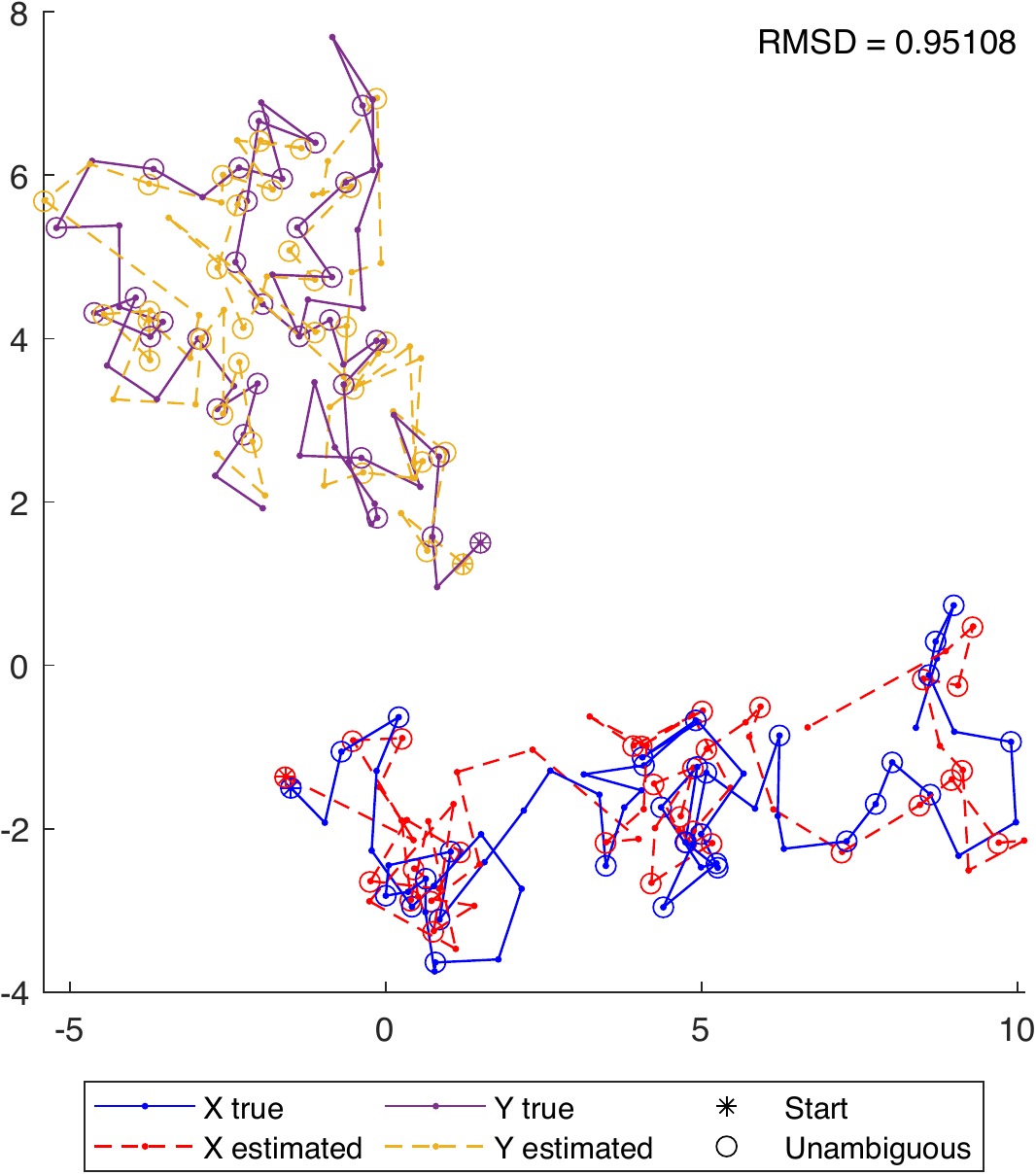}
    \caption{$\varepsilon=0.50$}
    \end{subfigure}
    ~
    \begin{subfigure}{0.3\textwidth}
    \includegraphics[width=\textwidth]{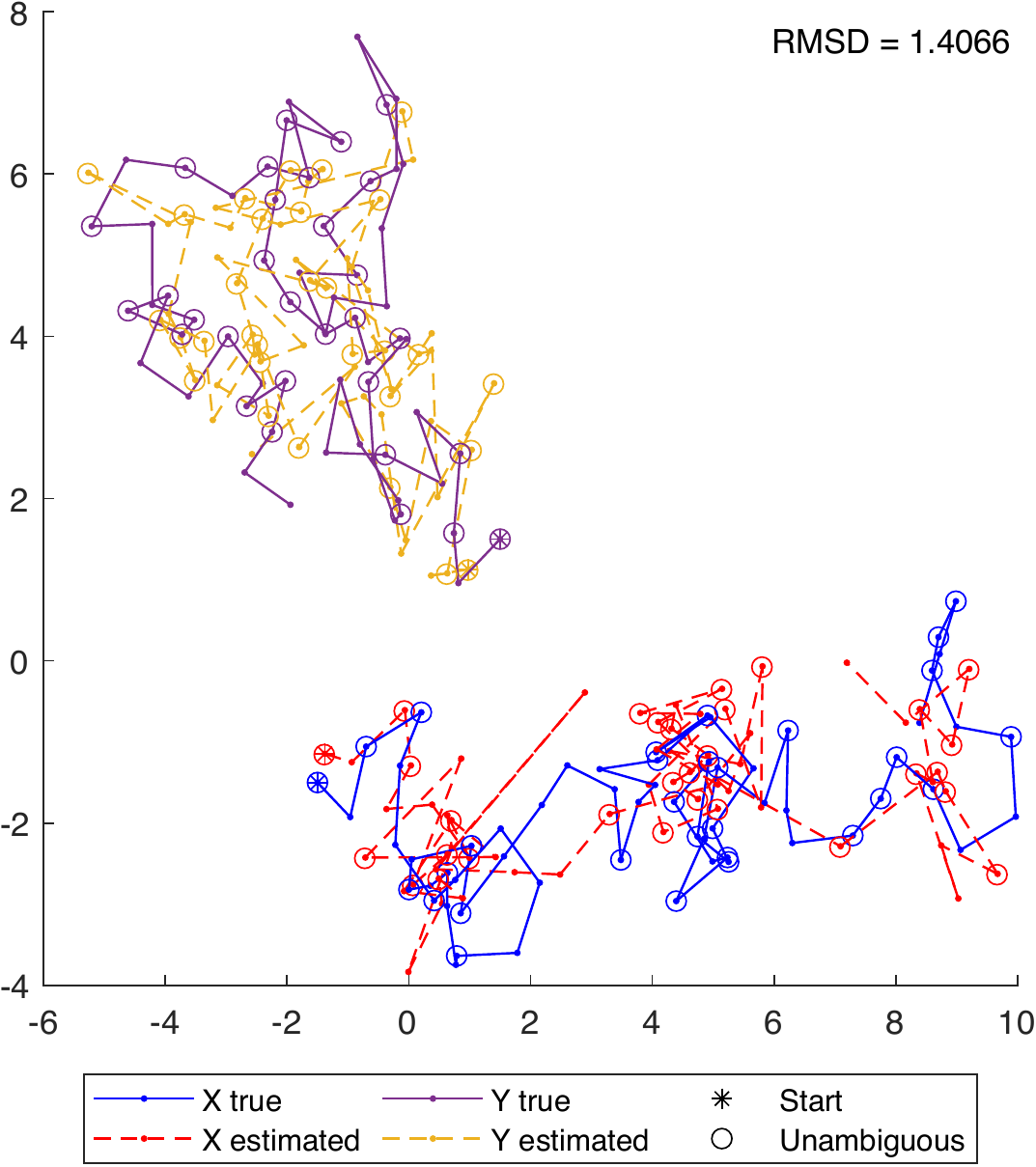}
    \caption{$\varepsilon=0.90$}
    \end{subfigure}
    \centering
    \caption{SNLC reconstructions, without the local optimization step.}
    \label{fig:reconstructions_without_local}
\end{figure}

Figure~\ref{fig:real_data_preprocessing} illustrates the preprocessing steps of the real dataset where loci with low contact counts are removed and the rest of the loci are partitioned into unambiguous and ambiguous. The total contact count for the $i$th locus is defined as the sum of all contacts where it participates:
\[T(i)\hspace{-3pt}=\hspace{-3pt}{\hspace{-5pt}}\sum_{j\in [n]} {\hspace{-6pt}}\left( c^A(i,j) \hspace{-2pt}+\hspace{-1pt}  c^P(i,j)  \hspace{-2pt} + \hspace{-1pt} c^P(i+n,j)\right){\hspace{-2pt}} +{\hspace{-8pt}} \sum_{j\in [2n]} {\hspace{-7pt}}\left(c^P(j,i)\hspace{-2pt}+\hspace{-2pt}c^U(i,j) \hspace{-2pt}+\hspace{-2pt} c^U(i \hspace{-1pt}+ \hspace{-1pt}n,j) \right).\]
Similarly, we define the unambiguity quotient as the proportion of $T(i)$ that consists of contacts where $x_i$ and $y_i$ could be distinguished:
\[\textit{UQ}(i)=\frac{1}{T(i)}\left(\sum_{j\in [n]} \left( c^P(i,j) + c^P(i+n,j)\right)  \hspace{-1pt} + \hspace{-6pt} \sum_{j\in [2n]} \left(c^U(i,j) + c^U(i + n,j) \right)\right).\]

\begin{figure}[h]
    \begin{subfigure}{0.3\textwidth}
    \includegraphics[width=\textwidth]{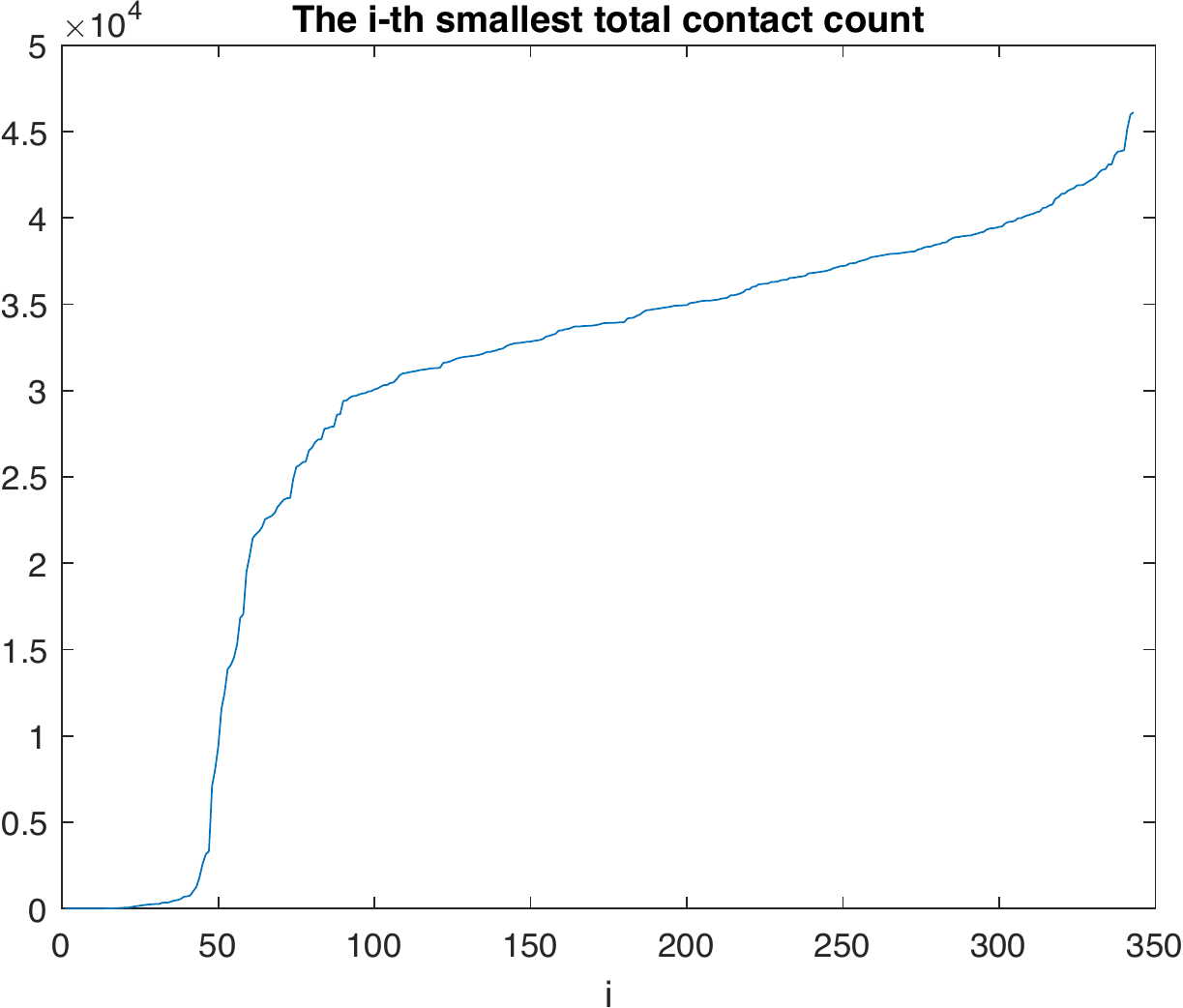}
    \caption{}
    \end{subfigure}
    ~
    \begin{subfigure}{0.3\textwidth}
    \includegraphics[width=\textwidth]{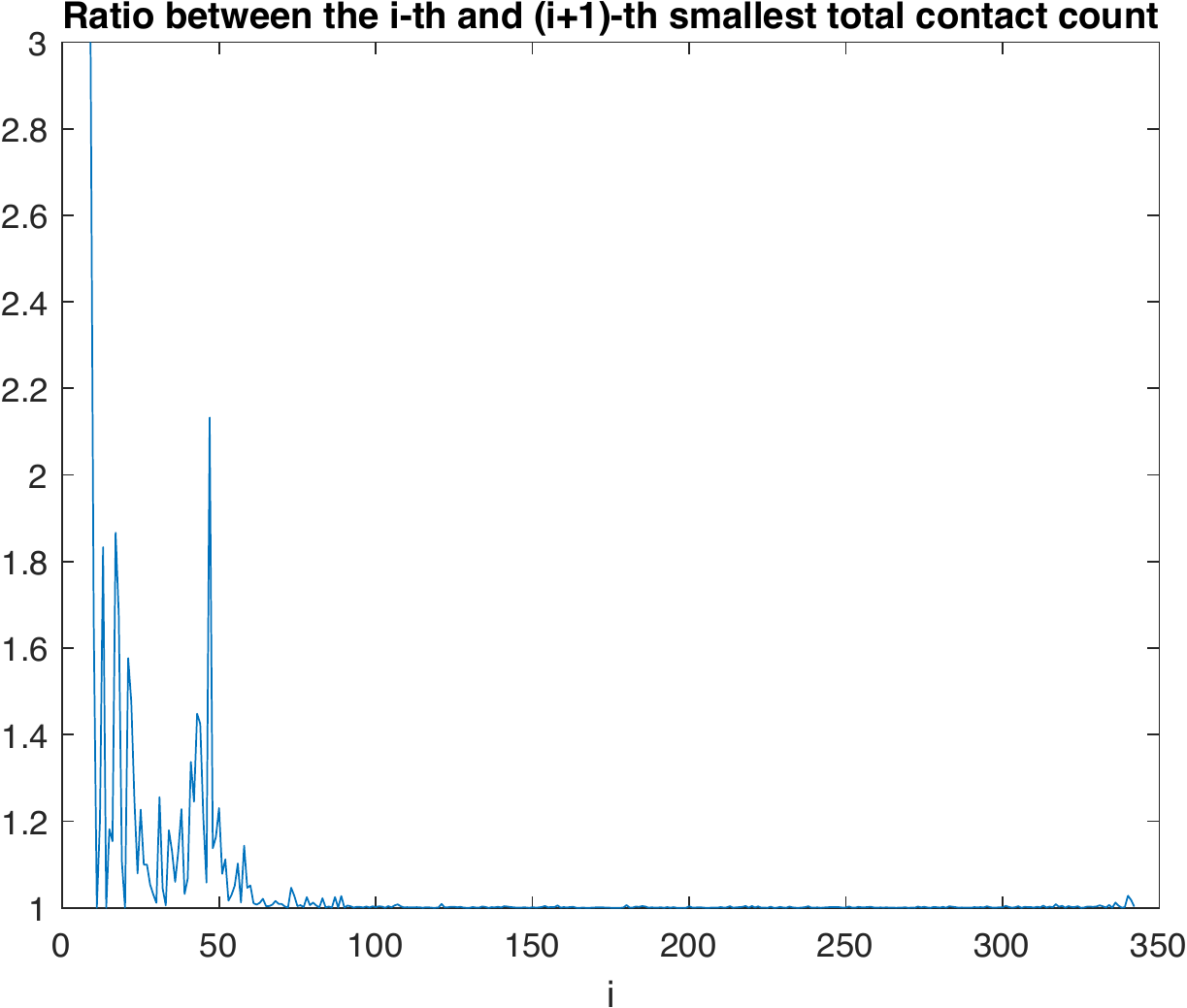}
    \caption{}
    \end{subfigure}
    ~
    \begin{subfigure}{0.3\textwidth}
    \includegraphics[width=\textwidth]{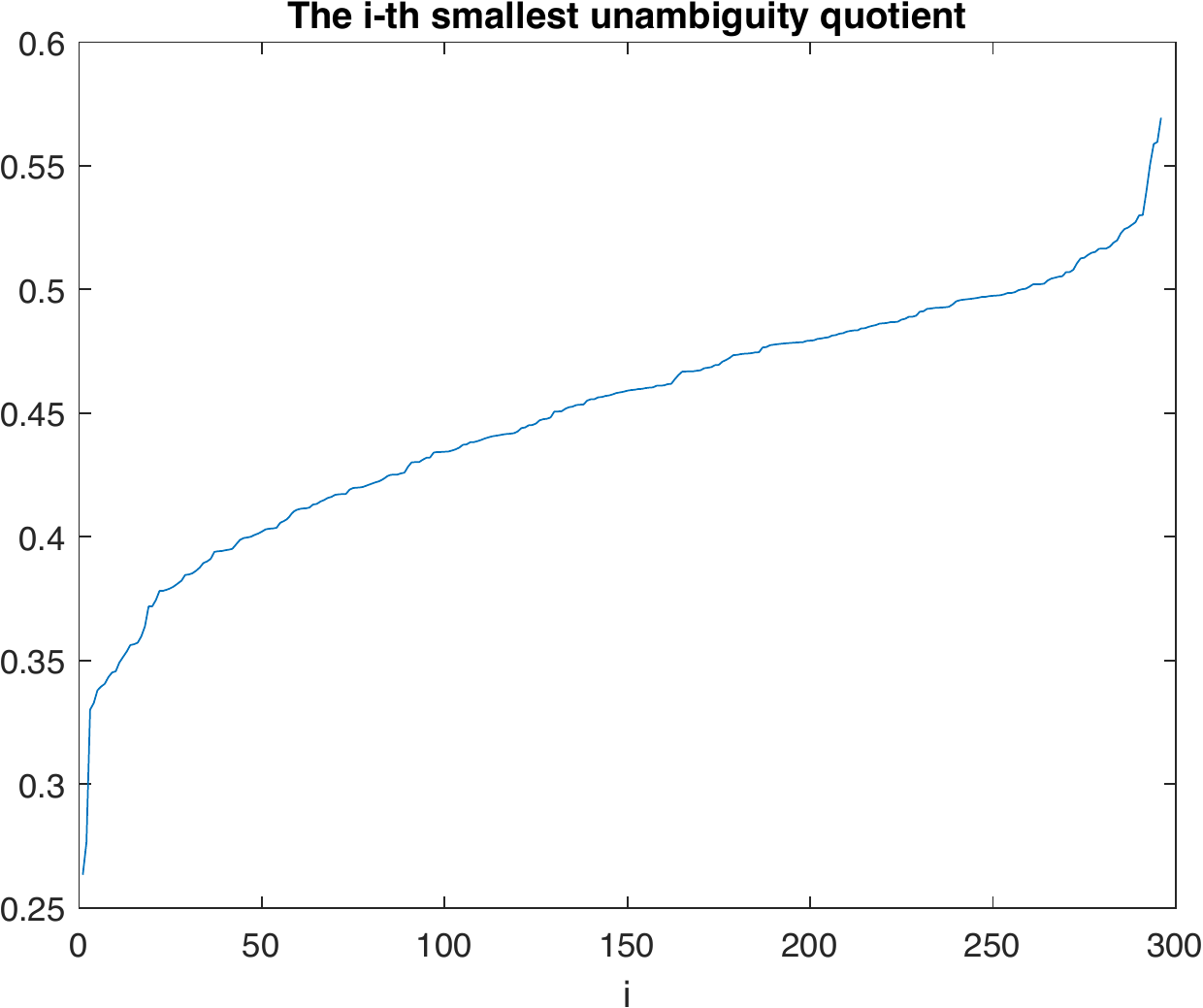}
    \caption{}
    \end{subfigure}
    \centering
    \caption{(a) Total contact counts sorted in increasing order. (b) Ratios between total contact counts. The peak corresponding to the ratio between the 48th and the 47th smallest count is used as a motivation for excluding the 47 loci with smallest total contact from the analysis. (c) Unambiguity quotients for each of the remaining 296 loci, sorted in increasing order. We consider a locus as ambiguous if this ratio is less than 0.4; otherwise, we consider it as unambiguous.}
    \label{fig:real_data_preprocessing}
\end{figure}

To obtain the setting of our paper where loci are partitioned into unambiguous or ambiguous, we reassign the contact counts of $\tilde{C}^U$ $\tilde{C}^P$ and $\tilde{C}^A$ of the Patski dataset according to whether a locus is unambiguous or ambiguous.
For $i,j\in U$, we define
\begin{flalign*}
&c^U_{i,j} \hspace{-2pt}= \hspace{-2pt}\tilde{c}^U_{i,j}  \hspace{-2pt}+ \hspace{-1pt}\tilde{c}^P_{i,j} \frac{\tilde{c}^U_{i,j}}{\tilde{c}^U_{i,j} \hspace{-2pt}+ \hspace{-1pt}\tilde{c}^U_{i,j+n}} \hspace{-1pt} +  \hspace{-1pt}\tilde{c}^P_{j,i}\frac{\tilde{c}^U_{i,j}}{\tilde{c}^U_{i,j} \hspace{-1pt}+ \hspace{-1pt}\tilde{c}^U_{i+n,j}}  \hspace{-1pt}+ \hspace{-1pt} \tilde{c}^A_{i,j} \frac{\tilde{c}^U_{i,j}}{\tilde{c}^U_{i,j} \hspace{-3pt}+ \hspace{-1pt}\tilde{c}^U_{i,j+n} \hspace{-3pt}+ \hspace{-1pt} \hspace{-1pt}\tilde{c}^U_{i+n,j} \hspace{-2pt}+ \hspace{-1pt}\tilde{c}^U_{i+n,j+n}},\\[0.5em]
&c^U_{i,j+n}\hspace{-2pt}=\hspace{-2pt}\tilde{c}^U_{i,j+n} + \tilde{c}^P_{i,j} \frac{\tilde{c}^U_{i,j+n}}{\tilde{c}^U_{i,j}+\tilde{c}^U_{i,j+n}} + \tilde{c}^P_{j+n,i}\frac{\tilde{c}^U_{i,j+n}}{\tilde{c}^U_{i,j+n}+\tilde{c}^U_{i+n,j+n}}+ \\ 
&{\hspace{35pt}}+\tilde{c}^A_{i,j} \frac{\tilde{c}^U_{i,j+n}}{\tilde{c}^U_{i,j}+\tilde{c}^U_{i,j+n}+\tilde{c}^U_{i+n,j}+\tilde{c}^U_{i+n,j+n}},\\[0.5em]
&c^U_{i+n,j}\hspace{-2pt}=\hspace{-2pt}\tilde{c}^U_{i+n,j} + \tilde{c}^P_{i+n,j} \frac{\tilde{c}^U_{i+n,j}}{\tilde{c}^U_{i+n,j}+\tilde{c}^U_{i+n,j+n}} + \tilde{c}^P_{j,i}\frac{\tilde{c}^U_{i+n,j}}{\tilde{c}^U_{i,j}+\tilde{c}^U_{i+n,j}}+ \\ 
&{\hspace{35pt}}+\tilde{c}^A_{i,j} \frac{\tilde{c}^U_{i+n,j}}{\tilde{c}^U_{i,j}+\tilde{c}^U_{i,j+n}+\tilde{c}^U_{i+n,j}+\tilde{c}^U_{i+n,j+n}},\\[0.5em]
&c^U_{i+n,j+n}\hspace{-2pt}=\hspace{-2pt}\tilde{c}^U_{i+n,j+n} + \tilde{c}^P_{i+n,j} \frac{\tilde{c}^U_{i+n,j+n}}{\tilde{c}^U_{i+n,j}+\tilde{c}^U_{i+n,j+n}} + \tilde{c}^P_{j+n,i}\frac{\tilde{c}^U_{i+n,j+n}}{\tilde{c}^U_{i,j+n}+\tilde{c}^U_{i+n,j+n}} +\\
&{\hspace{47pt}}+\tilde{c}^A_{i,j} \frac{\tilde{c}^U_{i+n,j+n}}{\tilde{c}^U_{i,j}+\tilde{c}^U_{i,j+n}+\tilde{c}^U_{i+n,j}+\tilde{c}^U_{i+n,j+n}}.
\end{flalign*}
For $i\in U, j\in A$, we define
\begin{flalign*}
&c^P_{i,j}\hspace{-2pt}=\hspace{-2pt}\tilde{c}^U_{i,j} + \tilde{c}^U_{i,j+n} + \tilde{c}^P_{i,j} + \tilde{c}^P_{j,i}\frac{\tilde{c}^U_{i,j}}{\tilde{c}^U_{i,j}+\tilde{c}^U_{i+n,j}} + \tilde{c}^P_{j+n,i}\frac{\tilde{c}^U_{i,j+n}}{\tilde{c}^U_{i,j+n}+\tilde{c}^U_{i+n,j+n}} +\\ 
&{\hspace{25pt}}+\tilde{c}^A_{i,j} \frac{\tilde{c}^P_{i,j}}{\tilde{c}^P_{i,j}+\tilde{c}^P_{i+n,j}},\\[0.5em]
&c^P_{i+n,j}\hspace{-2pt}=\hspace{-2pt}\tilde{c}^U_{i+n,j} \hspace{-1pt}+\hspace{-1pt} \tilde{c}^U_{i+n,j+n} \hspace{-1pt}+\hspace{-1pt} \tilde{c}^P_{i+n,j} \hspace{-1pt}+\hspace{-1pt} \tilde{c}^P_{j,i}\frac{\tilde{c}^U_{i+n,j}}{\tilde{c}^U_{i,j}+\tilde{c}^U_{i+n,j}} \hspace{-1pt}+\hspace{-1pt} \tilde{c}^P_{j+n,i}\frac{\tilde{c}^U_{i+n,j+n}}{\tilde{c}^U_{i,j+n}\hspace{-1pt}+\hspace{-1pt}\tilde{c}^U_{i+n,j+n}} +\\
&{\hspace{35pt}}+\tilde{c}^A_{i,j} \frac{\tilde{c}^P_{i+n,j}}{\tilde{c}^P_{i,j}+\tilde{c}^P_{i+n,j}}.
\end{flalign*}
Finally, for $i,j\in A$, we define\\[-0.5em]
\begin{flalign*}
c^A_{i,j} = \tilde{c}^U_{i,j}+\tilde{c}^U_{i,j+n}+\tilde{c}^U_{i+n,j}+\tilde{c}^U_{i+n,j+n} + \tilde{c}^P_{i,j}+\tilde{c}^P_{i+n,j}+\tilde{c}^P_{j,i}+\tilde{c}^P_{j+n,i} + 
\tilde{c}^A_{i,j}.&&
\end{flalign*}

In Figure~\ref{fig:reconstructed_contacts} in Appendix, the experimental contact counts from the Patski dataset are compared with the contact counts from the SNLC reconstruction.

\begin{figure}[h]
    \centering
    \begin{subfigure}{0.4\textwidth}
    \includegraphics[width=\textwidth]{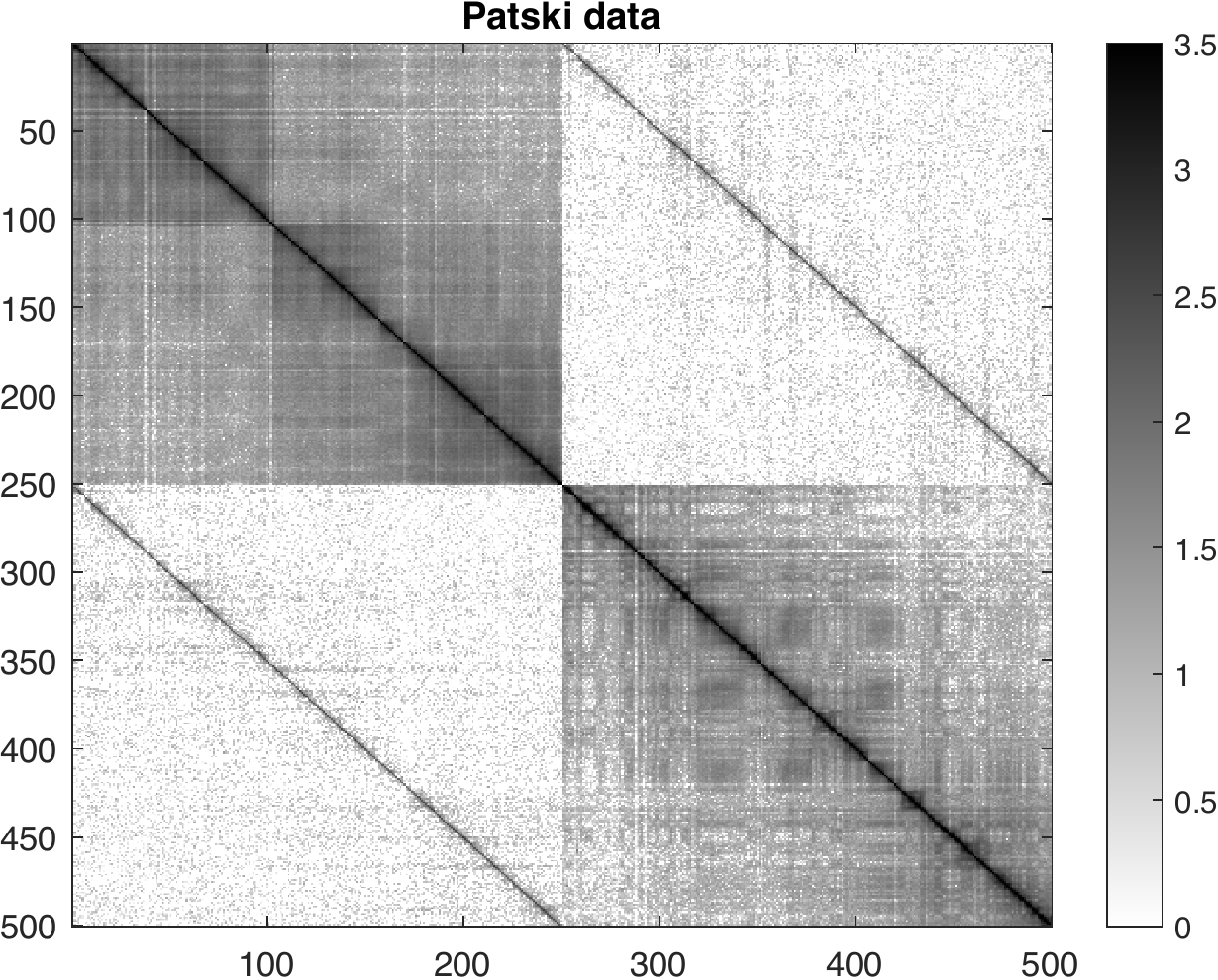}
    \caption{}
    \end{subfigure}
    \quad
    \begin{subfigure}{0.4\textwidth}
    \includegraphics[width=\textwidth]{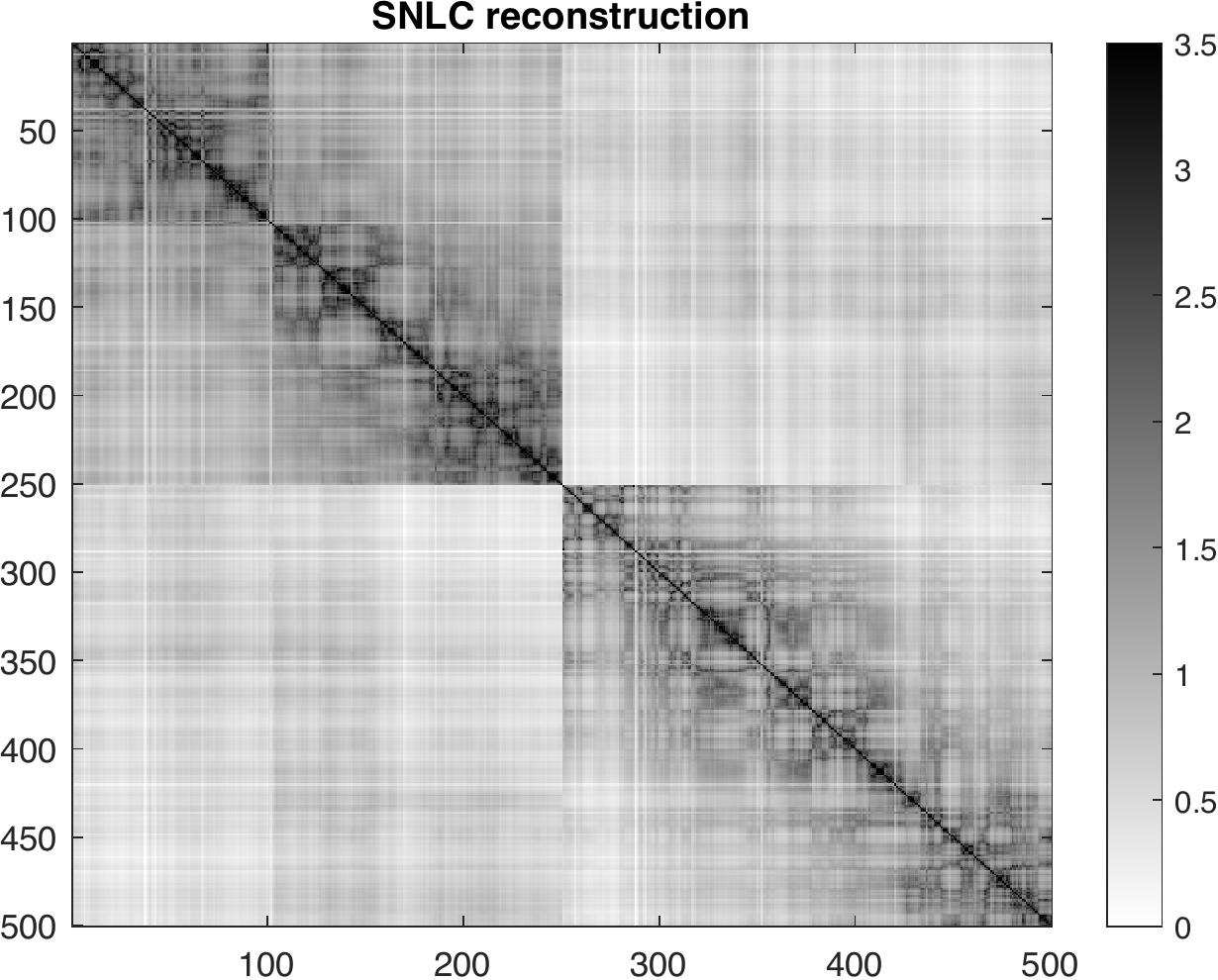}
    \caption{}
    \end{subfigure}\\[1em]
    \begin{subfigure}{0.4\textwidth}
    \includegraphics[width=\textwidth]{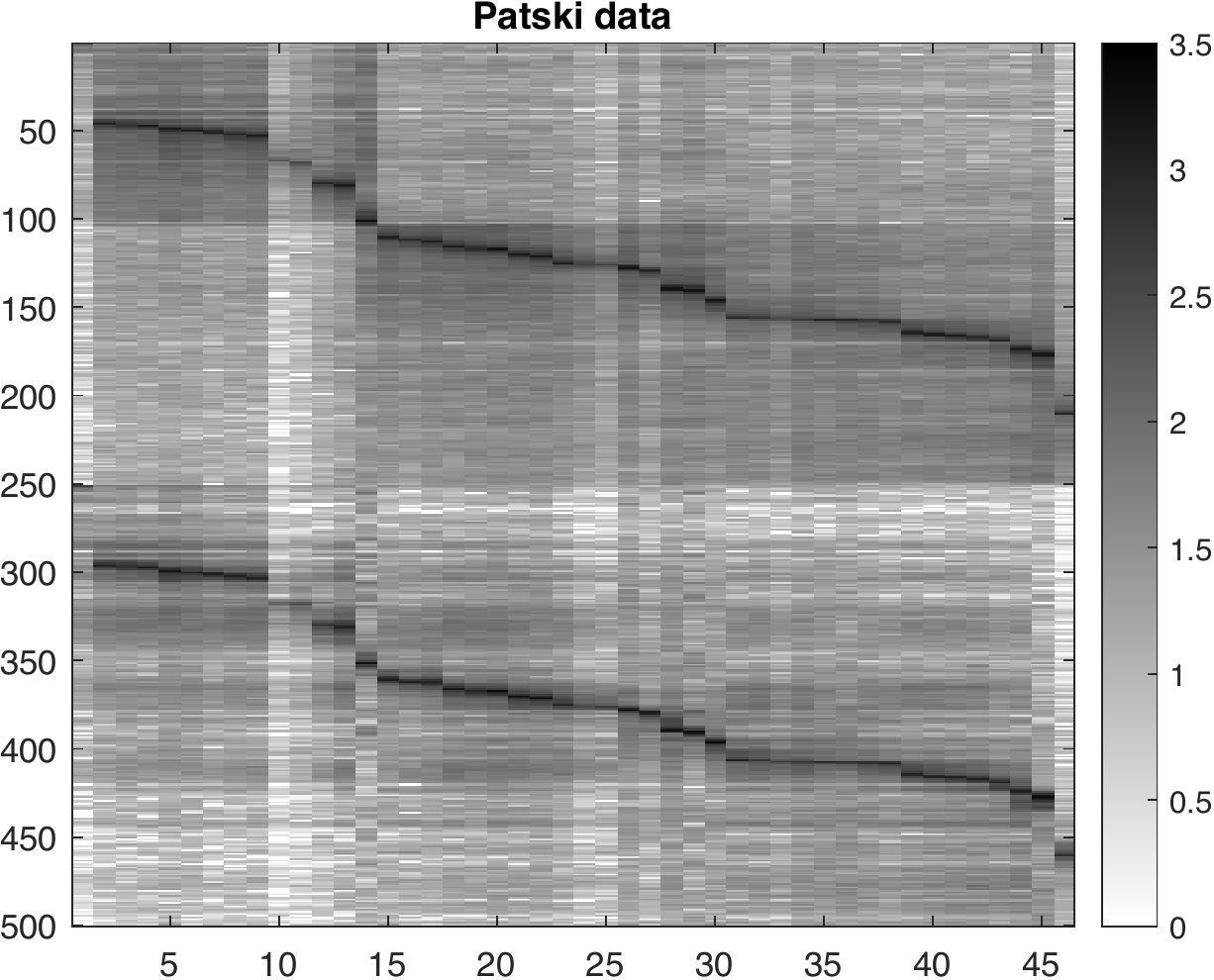}
    \caption{}
    \end{subfigure}
    \quad
    \begin{subfigure}{0.4\textwidth}
    \includegraphics[width=\textwidth]{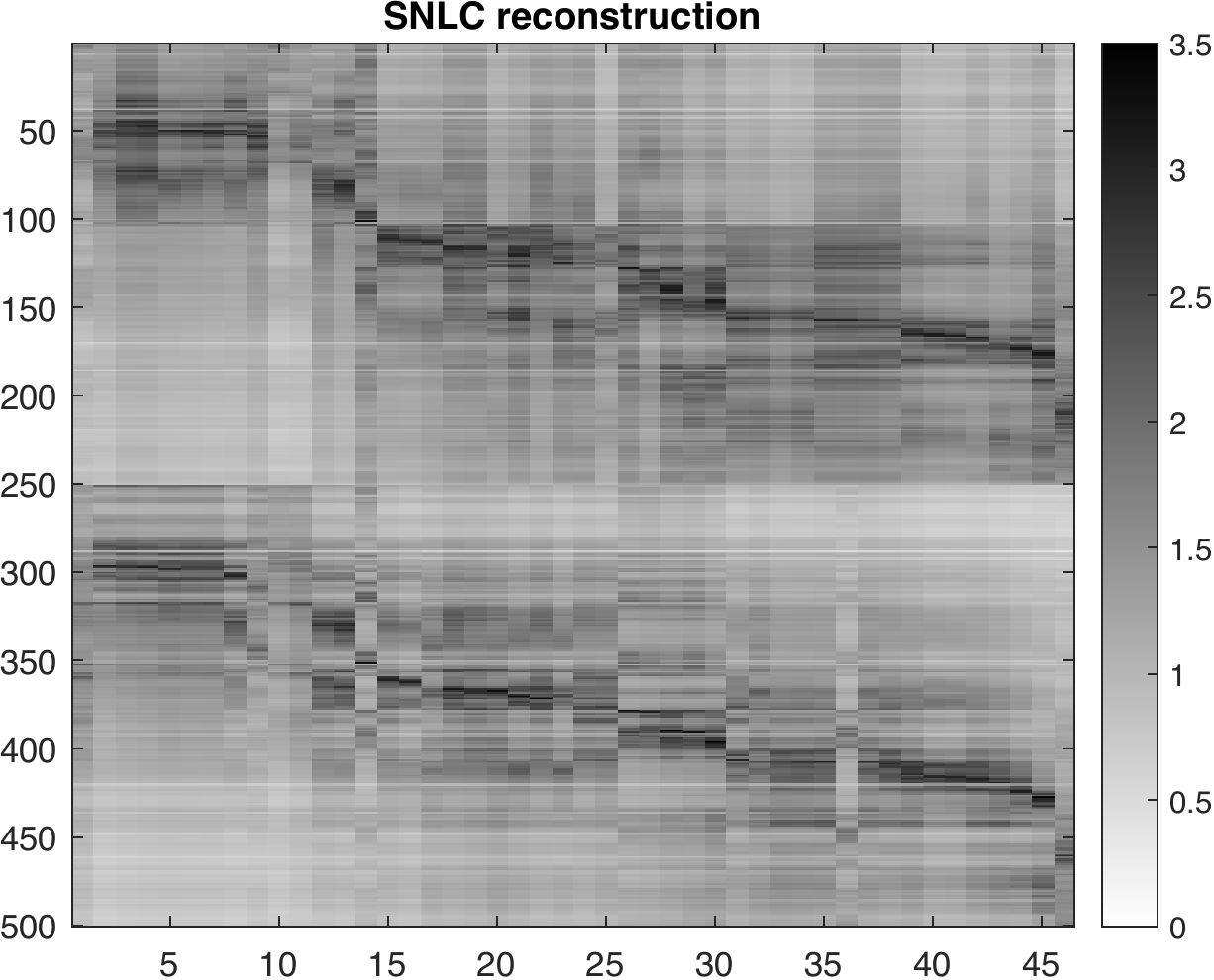}
    \caption{}
    \end{subfigure}\\[1em]
    \begin{subfigure}{0.4\textwidth}
    \includegraphics[width=\textwidth]{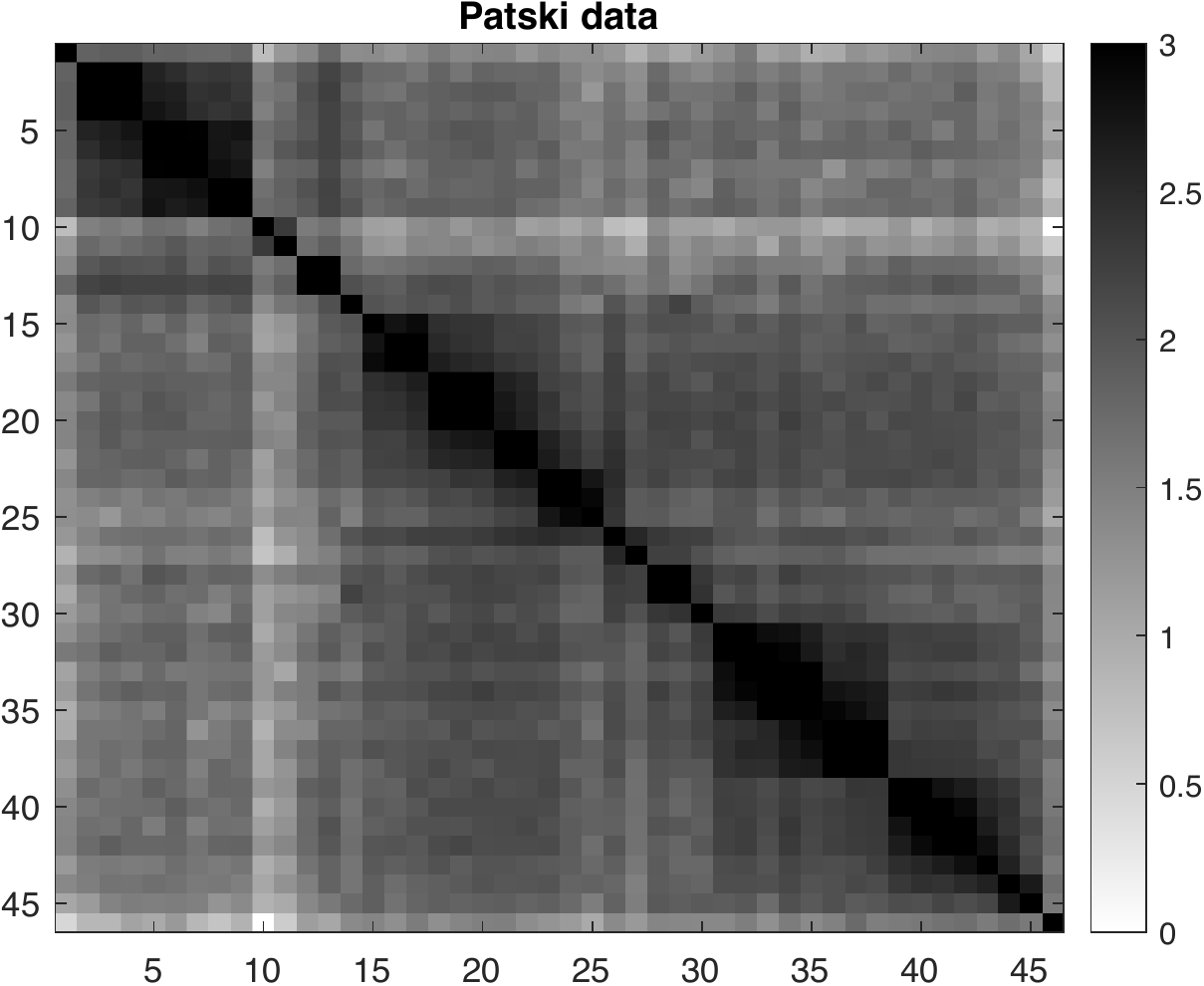}
    \caption{}
    \end{subfigure}
    \quad
    \begin{subfigure}{0.4\textwidth}
    \includegraphics[width=\textwidth]{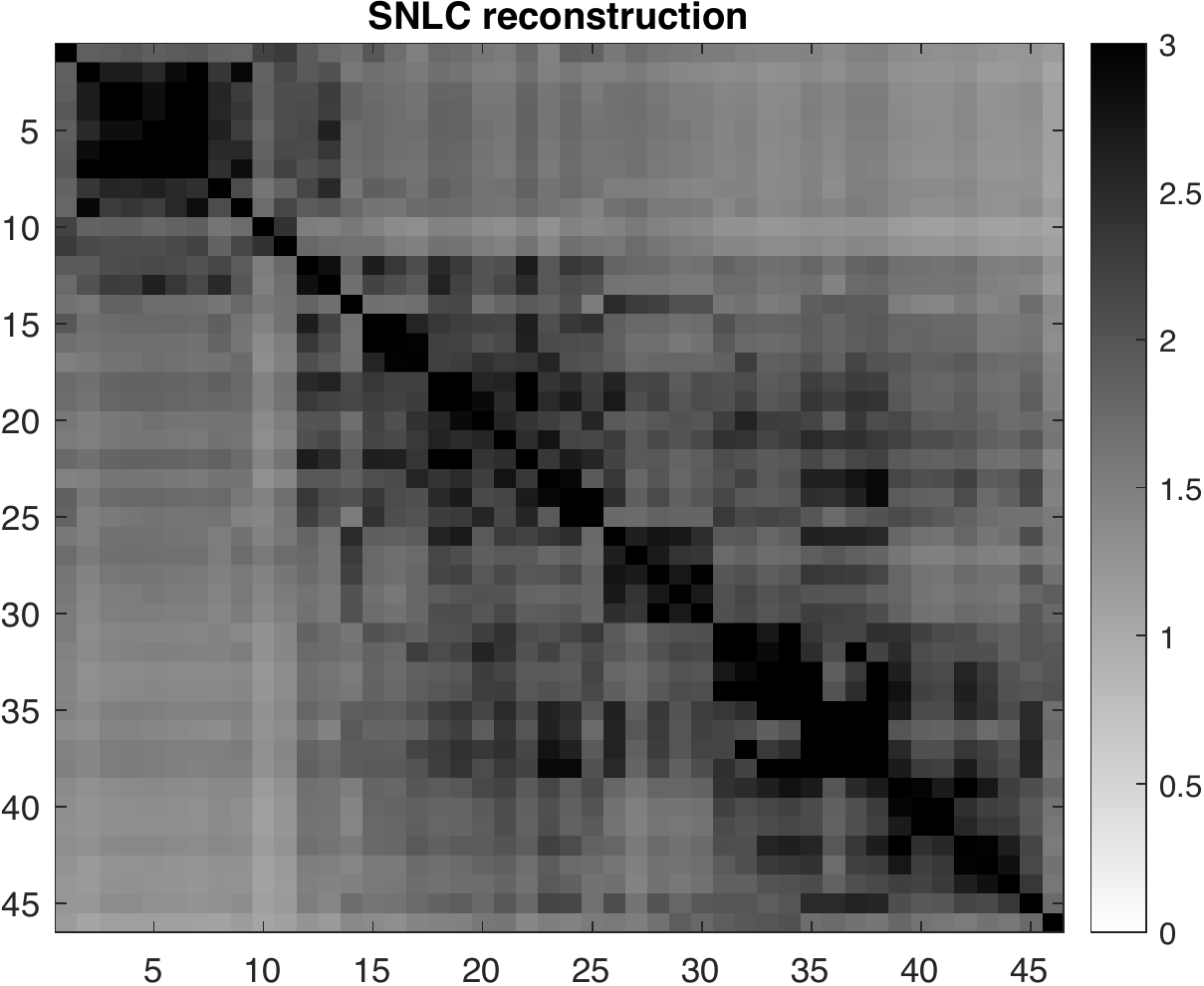}
    \caption{}
    \end{subfigure}
    \caption{Logarithmic heat maps for the reassigned contact count matrices obtained from the original Patski dataset and from the SNLC reconstruction: (a) and (b) $C^U$; (c) and (d) $C^P$; (e) and (f) $C^A$. The axis labels correspond to the 500 unambiguous beads, and the 46 ambiguous loci.}
    \label{fig:reconstructed_contacts}
\end{figure}

Figure~\ref{fig:norms_of_imaginary_parts} shows how the max-norm of the imaginary part of the solutions varies between different instances of the system \eqref{poly_eq_noisy} used for the reconstruction in Figure~\ref{fig:examples_of_reconstructions}(b), and for the reconstruction from the Patski data in Figure~\ref{fig:real_reconstructions}. A complete set of figures for these two datasets can be found in the Github repository. Taken together, the figures indicate that a max-norm of 0.15 was an appropriate threshold for approximate realness for both data sets, in the sense that it is low enough to single out solutions that have significantly smaller imaginary parts than the others, while also ensuring that it is possible to find an approximately real solution for each ambiguous locus.

\begin{figure}[h]
    \centering
    \begin{subfigure}{0.24\textwidth}
        \includegraphics[width=\textwidth]{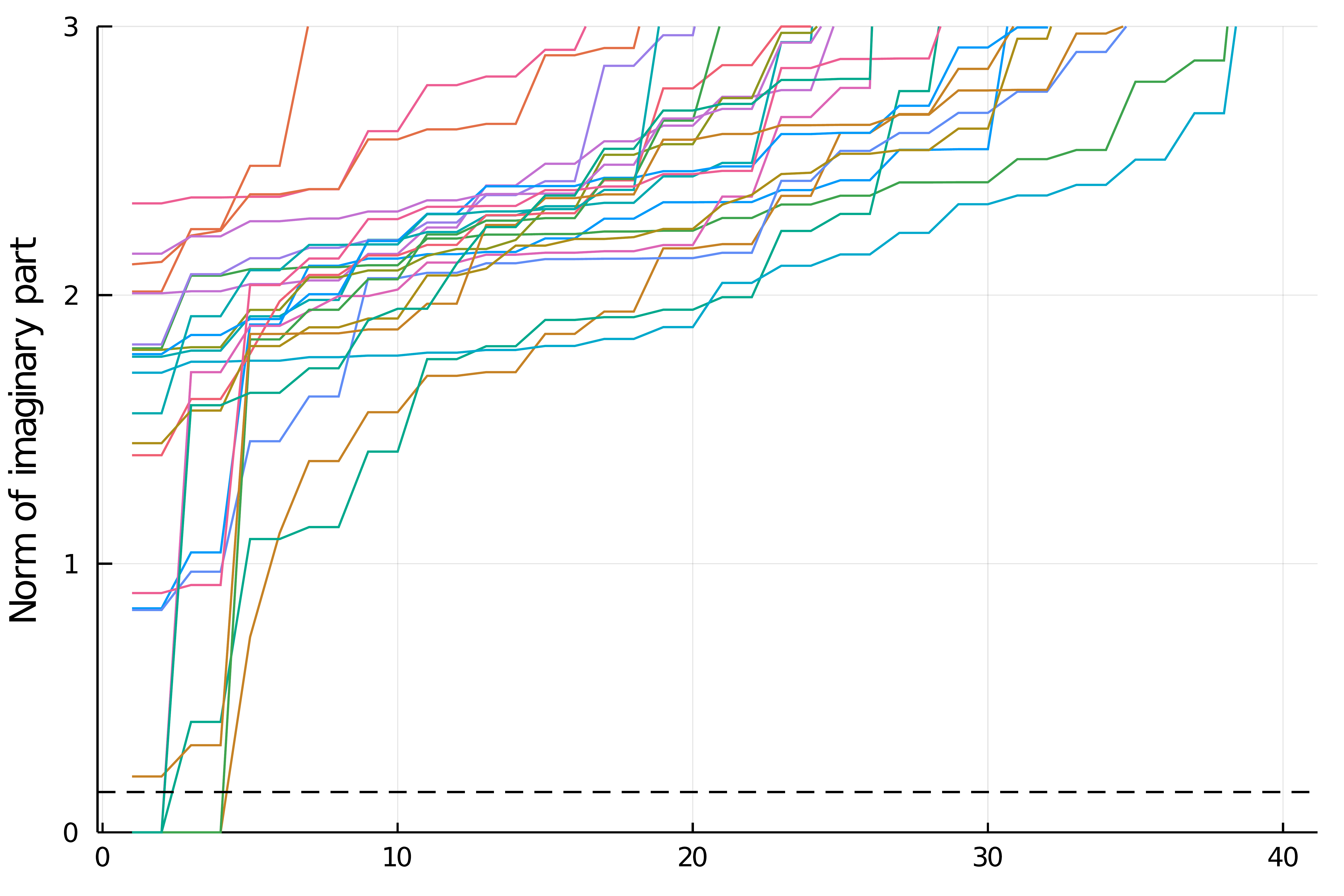}
        \caption{}
    \end{subfigure}
    \begin{subfigure}{0.24\textwidth}
        \includegraphics[width=\textwidth]{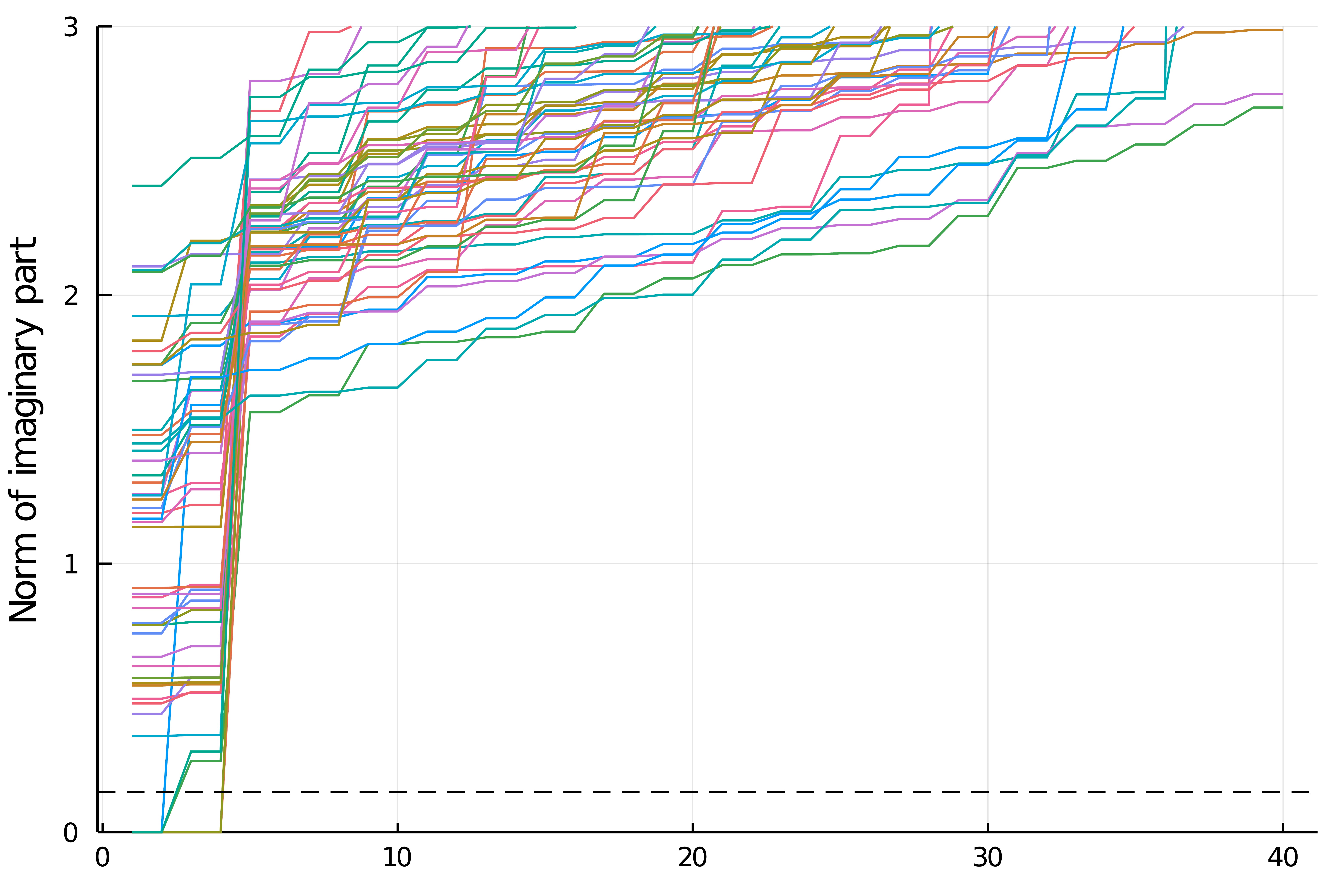}
        \caption{}
    \end{subfigure}
    \begin{subfigure}{0.24\textwidth}
        \includegraphics[width=\textwidth]{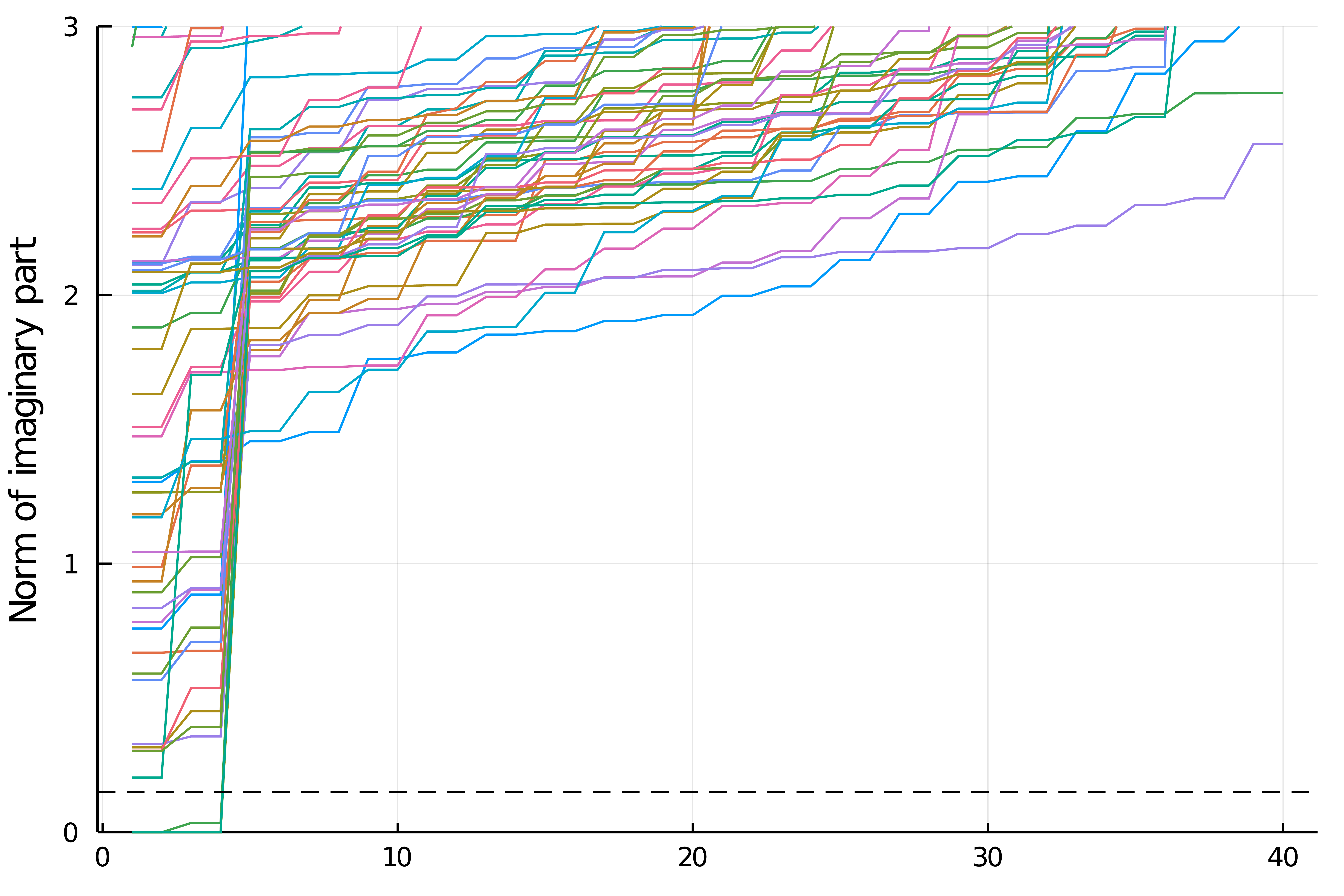}
        \caption{}
    \end{subfigure}
    \begin{subfigure}{0.24\textwidth}
        \includegraphics[width=\textwidth]{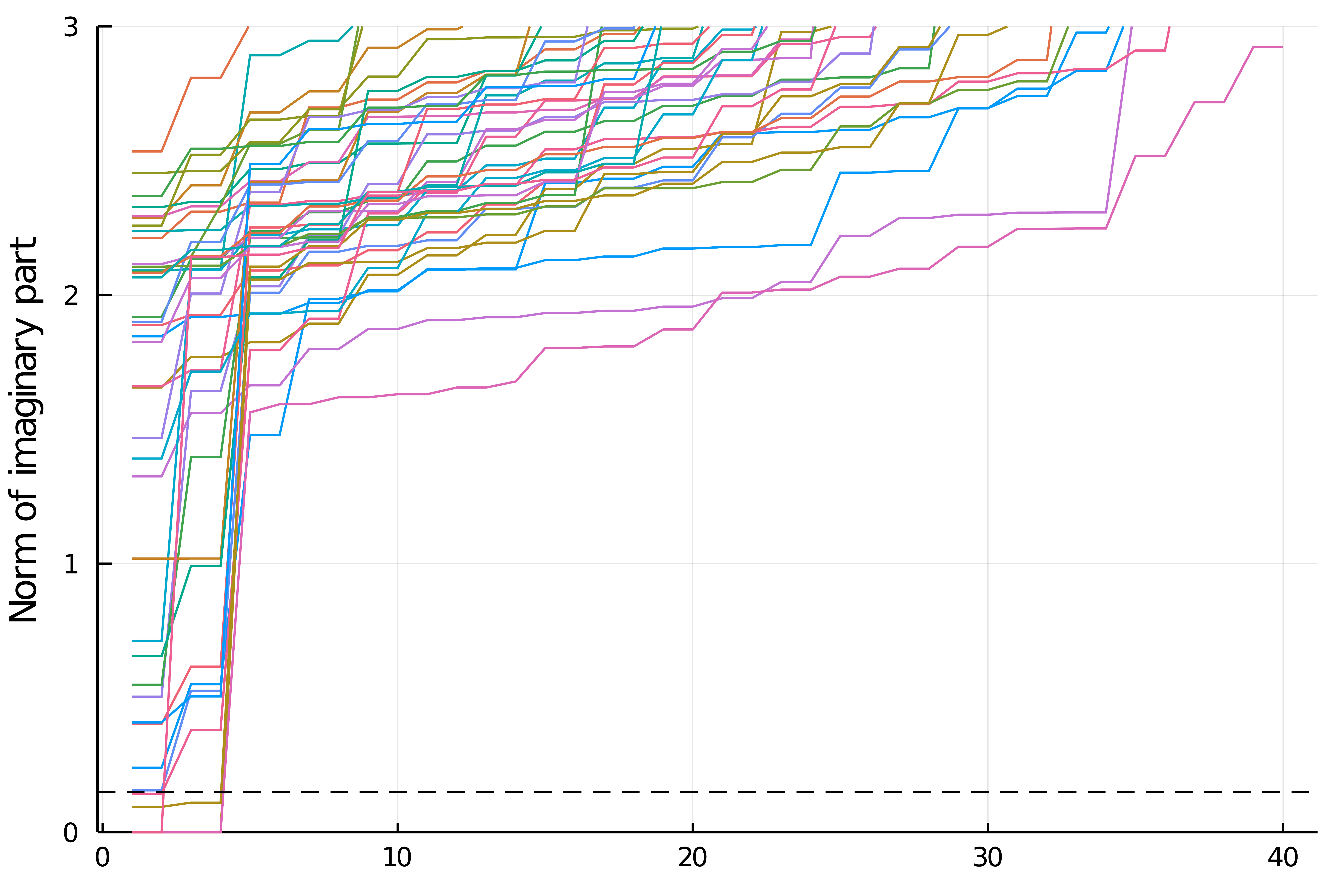}
        \caption{}
    \end{subfigure}

    \begin{subfigure}{0.24\textwidth}
        \includegraphics[width=\textwidth]{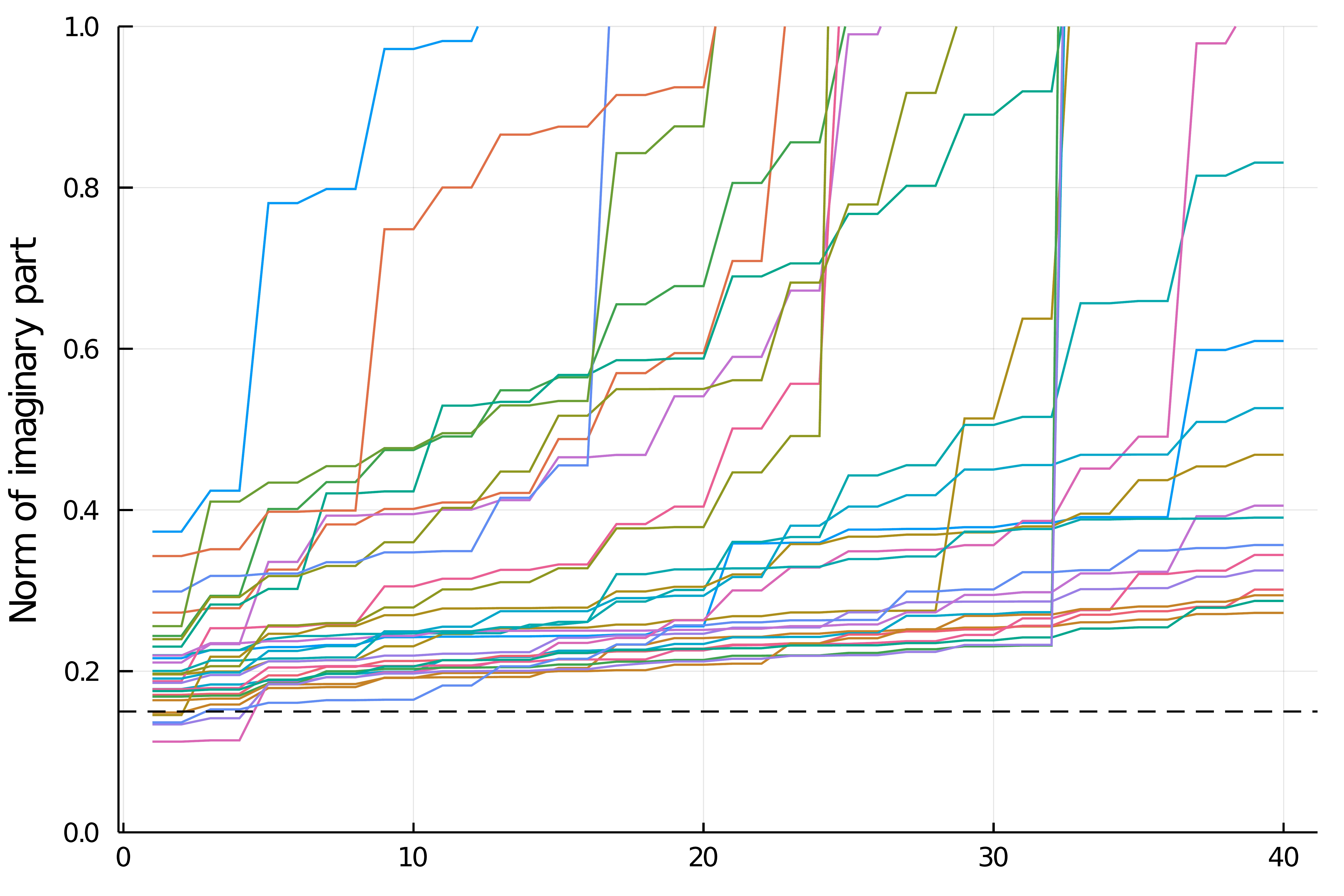}
        \caption{}
    \end{subfigure}
    \begin{subfigure}{0.24\textwidth}
        \includegraphics[width=\textwidth]{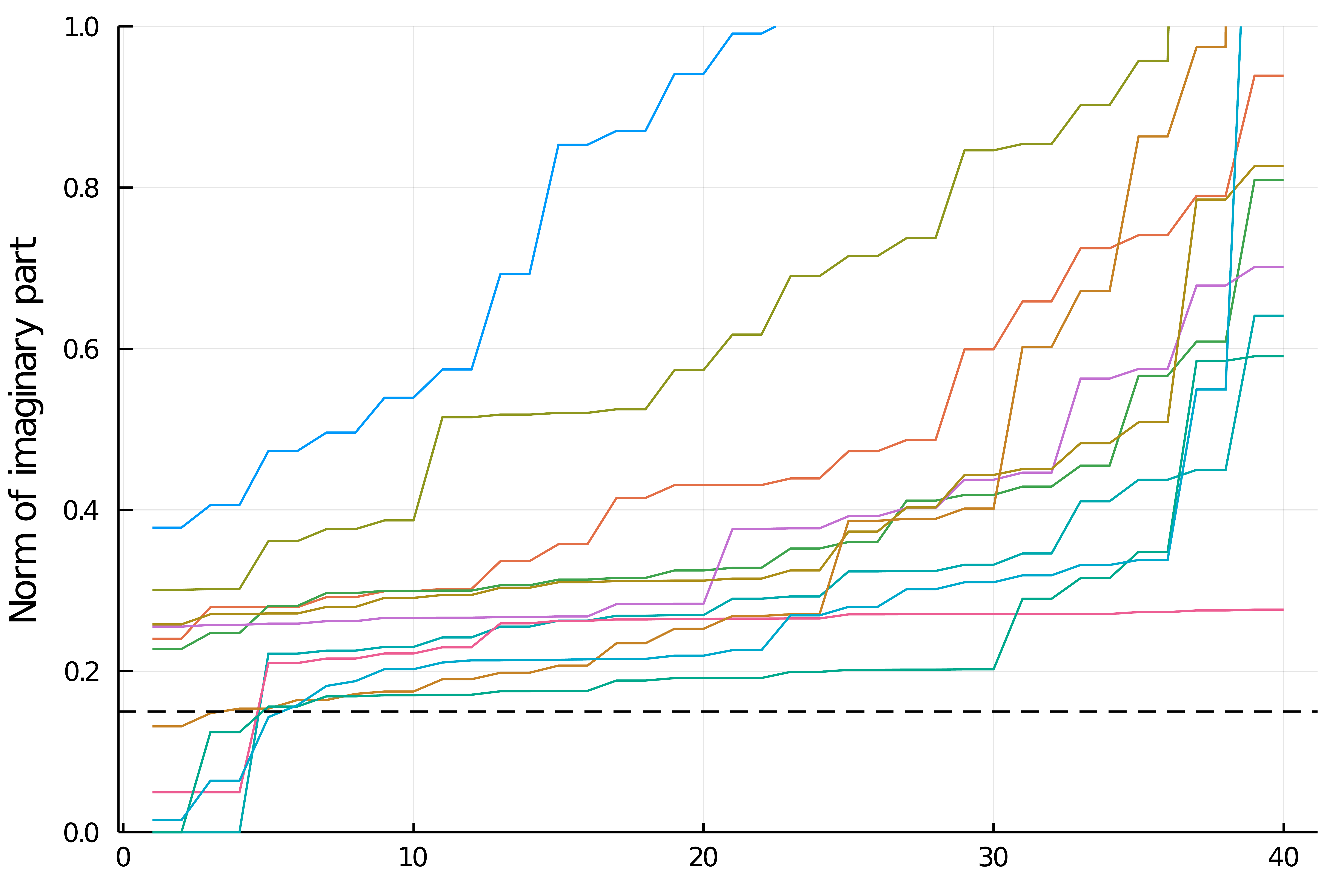}
        \caption{}
    \end{subfigure}
    \begin{subfigure}{0.24\textwidth}
        \includegraphics[width=\textwidth]{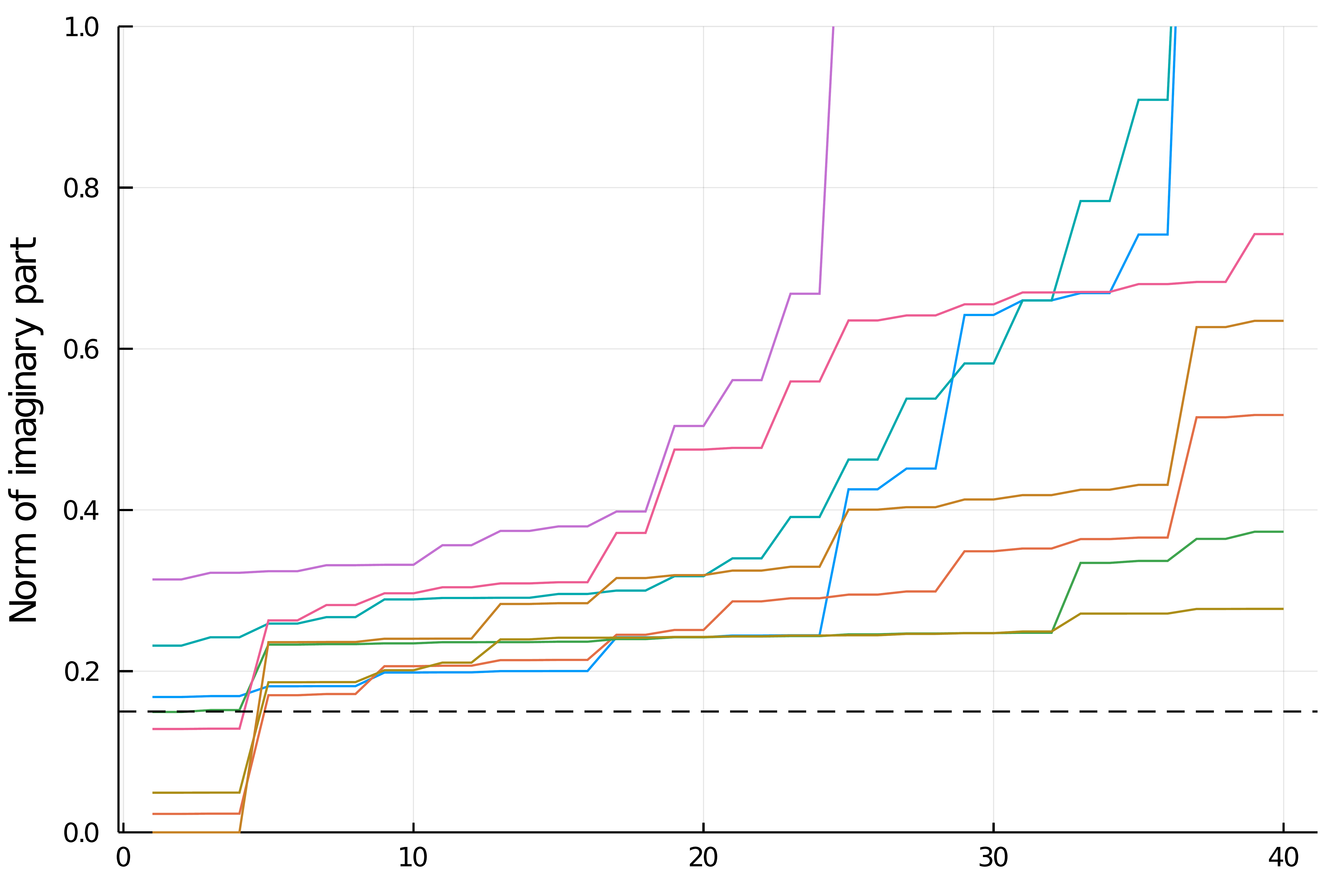}
        \caption{}
    \end{subfigure}
    \begin{subfigure}{0.24\textwidth}
        \includegraphics[width=\textwidth]{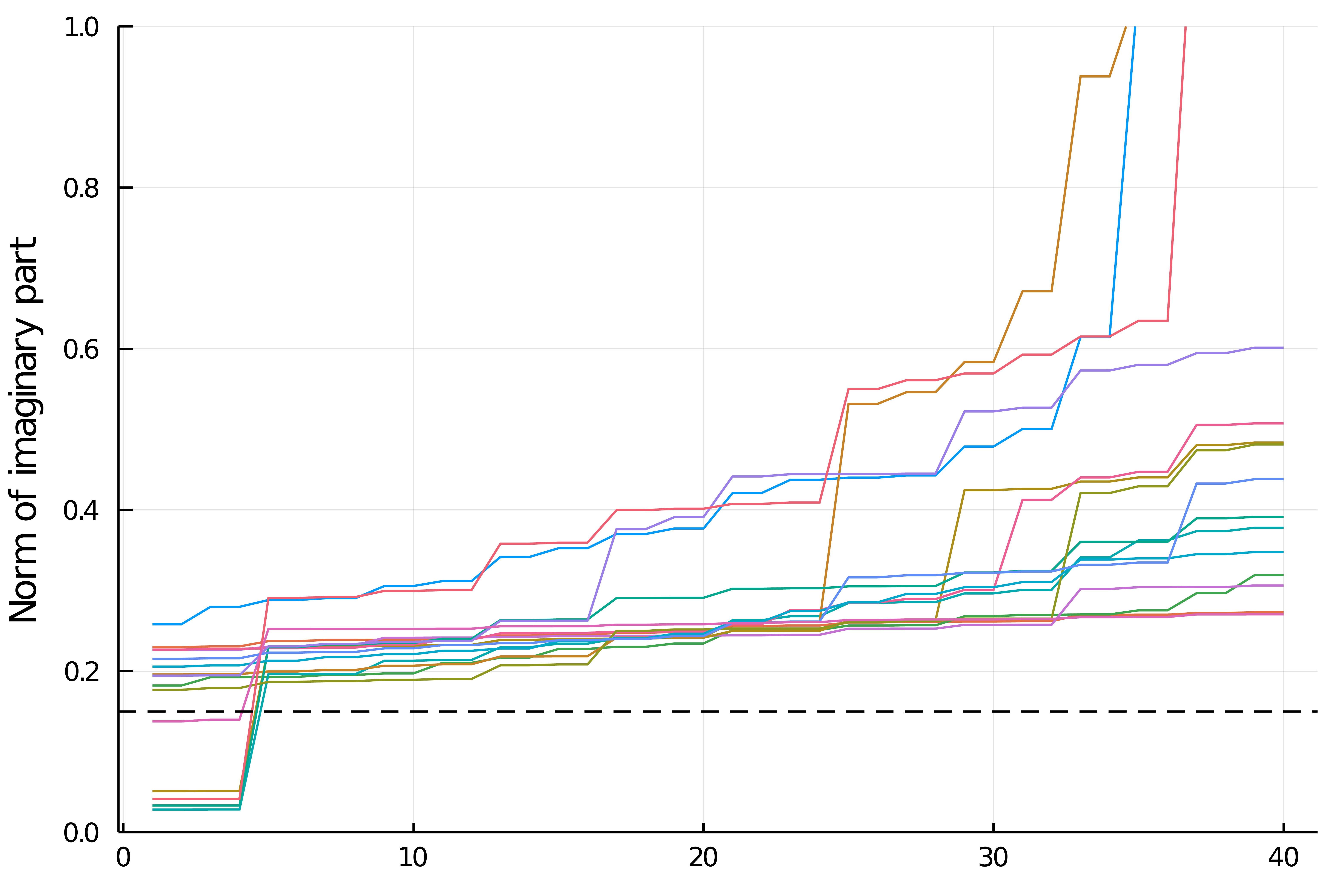}
        \caption{}
    \end{subfigure}   
    \caption{Max-norm of the imaginary parts encountered in the numerical algebraic geometry estimation of various loci. Each subfigure corresponds to an ambiguous locus: (a)--(d) correspond to the first four loci of the synthetic dataset used in Figure~\ref{fig:examples_of_reconstructions}(b); (e)--(h) correspond to the first four ambiguous loci of the Patski dataset. Each colored line corresponds to a specific choice of 6 unambiguous beads used in the estimation of the locus. Each line connects 40 points, that record the max-norm of the imaginary part of a solution (up to symmetry) found for the corresponding choice of 6 unambiguous beads. The dashed line at 0.15 corresponds to the choice of threshold for when a solution is considered approximately real. Similar figures for the rest of the ambiguous loci in the respective chromosome pairs can be found in the Github repository.}
    \label{fig:norms_of_imaginary_parts}
\end{figure}

\end{document}